\documentclass[10pt,journal,compsoc]{IEEEtran}
\ifCLASSOPTIONcompsoc
  \usepackage[nocompress]{cite}
\else
  \usepackage{cite}
\fi
\ifCLASSINFOpdf
\else
\fi

\usepackage{amsmath,amsbsy,amssymb,latexsym}
\usepackage{bm}
\usepackage{graphicx}
\usepackage{subfigure}
\usepackage{color, url}
\usepackage{mathtools}
\usepackage{amssymb}
\usepackage{amsthm}
\usepackage{amsmath}
\usepackage{algorithm}
\usepackage{algorithmic}
\usepackage[algo2e]{algorithm2e}

\newtheorem{theorem}{Theorem}
\newtheorem{lemma}{Lemma}
\newtheorem{definition}{Definition}

\newcommand{\myparatight}[1]{\noindent{\bf {#1}:}~}
\newenvironment{packeditemize}{\begin{list}{$\bullet$}{\setlength{\itemsep}{1pt}\addtolength{\labelwidth}{20pt}\setlength{\leftmargin}{\labelwidth}\setlength{\listparindent}{\parindent}\setlength{\parsep}{0pt}\setlength{\topsep}{0pt}}}{\end{list}}

\begin{document}
%
\title{Structure-based Sybil Detection in Social Networks via Local Rule-based Propagation}
%
%
%
%

\author{Binghui~Wang,~\IEEEmembership{Student~Member,~IEEE,}
        Jinyuan~Jia,~\IEEEmembership{Student~Member,~IEEE,}
        Le~Zhang,
        and~Neil~Zhenqiang~Gong,~\IEEEmembership{Member,~IEEE}
\IEEEcompsocitemizethanks{\IEEEcompsocthanksitem The authors are with the Department of Electrical 
and Computer Engineering, Iowa State University, Ames, IA, 50010.\protect\\
E-mail: \{binghuiw,jinyuan,lezhang,neilgong\}@iastate.edu.
\IEEEcompsocthanksitem Corresponding author: Neil~Zhenqiang~Gong (neilgong@iastate.edu) 
\IEEEcompsocthanksitem One of the 10 papers [49] fast tracked from INFOCOM'17. 

}
}

\markboth{IEEE Transactions on Network Science and Engineering, 2019}%
{Shell \MakeLowercase{\textit{et al.}}: Bare Demo of IEEEtran.cls for Computer Society Journals}
%



\IEEEtitleabstractindextext{%

\begin{abstract}
Social networks are known to be vulnerable to the so-called Sybil attack, in which an attacker maintains massive Sybils and uses them to perform various malicious activities. Therefore, Sybil detection in social networks is a basic security research problem.
 Structure-based methods have been shown to be promising at detecting Sybils. Existing  structure-based methods can be classified into two categories: Random Walk (RW)-based methods and Loop Belief Propagation (LBP)-based methods. 
RW-based methods cannot leverage labeled Sybils and labeled benign users simultaneously, which limits their detection accuracy, and/or  they are not robust to noisy labels. LBP-based methods  are not scalable, and they cannot guarantee convergence. 
In this work, we propose SybilSCAR, a novel structure-based method to detect Sybils in social networks. 
SybilSCAR is \underline{\emph{S}}calable, \underline{\emph{C}}onvergent, \underline{\emph{A}}ccurate, and \underline{\emph{R}}obust to label noise. 
 We first propose a framework to unify RW-based and LBP-based methods. Under our framework, these methods can be viewed as iteratively applying a (different) \emph{local rule} to every user, which propagates label information among a social graph. Second, we design a new local rule, which SybilSCAR iteratively applies to every user to detect Sybils. 
We compare SybilSCAR with state-of-the-art RW-based methods and LBP-based methods both theoretically and empirically. 
Theoretically, we show that, with proper parameter settings, SybilSCAR has a tighter asymptotical bound on the number of Sybils that are falsely accepted into a social network than existing structure-based methods. 
Empirically, we perform evaluation using both social networks with synthesized Sybils and a large-scale Twitter dataset (41.7M nodes and 1.2B edges) with real Sybils, and our results
 show that 
1) SybilSCAR is substantially more accurate and more robust to label noise than state-of-the-art RW-based methods; and 
2) SybilSCAR is more accurate and one order of magnitude more scalable than state-of-the-art LBP-based methods.
\end{abstract}

}

\maketitle

\IEEEdisplaynontitleabstractindextext

%
\IEEEpeerreviewmaketitle

\IEEEraisesectionheading{\section{Introduction}\label{sec:introduction}}

\IEEEPARstart{S}{ocial} networks are becoming more and more important and essential platforms for people to interact with each other, process information, and diffuse social influence, etc. 
For example, Facebook was reported to have 1.65 billion monthly active users as of April 2016~\cite{gao2012towards}, and it has become the third most visited website worldwide, just next to Google.com and YouTube.com, according to Alexa~\cite{alexaStat}.
However, it is well known that social networks are vulnerable to \emph{Sybil attacks}, in which attackers maintain a large number of Sybils, e.g., spammers, fake users, and compromised normal users.
For instance, 10\% of Twitter users were fake~\cite{Twittersybil}. 
Adversaries can leverage such Sybils to perform various malicious activities such as disrupting democratic election~\cite{election},  influencing financial market~\cite{stock}, distributing spams and phishing attacks~\cite{Thomas11}, as well as harvesting private user data~\cite{Bilge09}. 
Therefore, Sybil detection in social networks is an urgent research problem.

Indeed, this research problem has attracted increasing attention from multiple research communities including networking, security, and data mining. 
Among various methods, structure-based methods have demonstrated promising results, e.g., SybilRank~\cite{sybilrank} was deployed to detect a large amount of Sybils in Tuenti, the largest online social network in Spain. Most structure-based methods~\cite{Yu06,Yu08,Danezis09,Mohaisen11,sybilrank,Yang12-spam,integro,smartwalk,SybilWalk,zhang2016truetop,sybilbelief,sybilframe,robustspammer,gang} can be grouped into two categories: Random Walk (RW)-based methods and Loop Belief Propagation (LBP)-based methods. Given a training dataset, these methods iteratively propagate label information among the social graph to predict labels for users. RW-based methods implement the propagation using random walks, while  LBP-based methods implement the propagation using Loopy Belief Propagation~\cite{Pearl88}. 
RW-based methods~\cite{Yu06,Yu08,Danezis09,Mohaisen11,sybilrank,Yang12-spam,integro,smartwalk,SybilWalk} suffer from one or two major limitations: 1) they can only leverage either labeled benign users or labeled Sybils in the training dataset, but not both, which limits their detection accuracies;
and 2) they are not robust to label noise in the training dataset. The label of a user is noisy if the label is incorrect.  Label noise often exists in practice due to human mistakes when manually labeling users~\cite{Thomas11,wang2012social}. 
LBP-based methods~\cite{sybilbelief,sybilframe,robustspammer}  suffer from three major limitations: 
 1) they cannot guarantee convergence on real-world social networks; 2) they are not scalable; and 3) they do not have theoretically guaranteed performance.  The first limitation makes LBP-based methods sensitive to the number of iterations that the methods run. 

\myparatight{Our work} 
We propose a novel structure-based method, called SybilSCAR, to perform Sybil detection in social networks. 
SybilSCAR combines the advantages of RW-based methods and LBP-based methods, while overcoming their limitations. 
Specifically, SybilSCAR is \underline{\emph{S}}calable, \underline{\emph{C}}onvergent, \underline{\emph{A}}ccurate, and \underline{\emph{R}}obust to label noise. 

First, we propose a general framework to unify state-of-the-art RW-based and LBP-based methods. Under our framework, each structure-based method can be viewed as iteratively applying a \emph{local rule} to every node, which propagates label information from the training dataset to other nodes in the social network. A local rule updates a node's label information via combining the node's neighbors' label information and the prior knowledge that we know about the node. Although RW-based methods and LBP-based methods use very different mathematical foundations (i.e., RW vs. LBP), they can be viewed as applying different local rules under our framework. Our framework makes it possible to compare different methods in a unified way. Moreover, our framework provides new insights on how to design better structure-based methods. Specifically, designing better  structure-based methods reduces to designing better local rules. 

Second, we design a novel local rule that integrates the advantages of both RW-based methods and LBP-based methods, while overcoming their limitations. SybilSCAR iteratively applies our local rule to every user. 
Our local rule, like RW-based methods and LBP-based methods, leverages the \emph{homophily property} of social networks. 
Homophily means that two linked users share the same label with a high probability.
In our local rule, we associate a weight with each edge, which represents the probability that the two corresponding users have the same label. For a neighbor $v$ of $u$, our local rule models $v$'s influence (we call it \emph{neighbor influence}) to $u$'s label as the probability that $u$ is a Sybil, given $v$'s information alone. Our local rule combines neighbor influences and prior knowledge about a user in a multiplicative way to update knowledge about the user's label. Moreover, we linearize the multiplicative local rule in order to make SybilSCAR convergent.

Third, we evaluate SybilSCAR and compare it with state-of-the-art RW-based methods and LBP-based methods both theoretically and empirically.
Theoretically, we derive a bound on the number of Sybils that are accepted into a social network for SybilSCAR. Our bound is tighter than those of the existing methods. Moreover, we analyze the condition when SybilSCAR is guaranteed to converge. 
Empirically, we compare SybilSCAR with SybilRank~\cite{sybilrank}, a state-of-the-art RW-based method, and SybilBelief~\cite{sybilbelief}, a state-of-the-art LBP-based method, using 1) three real-world social networks with synthesized Sybils and 2) a large-scale Twitter dataset (41.7M users and 1.2B edges) with real Sybils.
Our empirical results demonstrate that 1) SybilSCAR achieves better detection accuracies than  SybilRank and SybilBelief, 2) SybilSCAR is robust to larger label noise than SybilRank, and is as robust as SybilBelief; 3) SybilSCAR is as space and time efficient as SybilRank, but is several times more space efficient and one order of magnitude more time efficient than SybilBelief; 4) SybilSCAR and SybilRank are convergent, but SybilBelief is not. For instance, in the large Twitter dataset, among the top-10K users that are predicted to be most likely Sybils by SybilRank, SybilBelief, and SybilSCAR, 0.33\%, 77.5\%, and 95.8\% of them are real Sybils, respectively.

In summary, our key contributions are as follows:

\begin{packeditemize}
\item We propose SybilSCAR, a novel structure-based methods, to detect Sybils in social networks. 
SybilSCAR is convergent, scalable, robust to label noise, and more accurate than existing methods. 

\item We propose a local rule-based framework to unify state-of-the-art RW-based methods and LBP-based methods. 
Under our framework, we design a novel local rule that is the key component of SybilSCAR.

\item We evaluate SybilSCAR both theoretically and empirically, and compare it with a state-of-the-art RW-based method and a state-of-the-art LBP-based method. 
Our theoretical results show that SybilSCAR has a tighter bound on the number of Sybils that are falsely accepted into a social network than existing methods. 
 Our empirical results on multiple social network datasets demonstrate that SybilSCAR significantly outperforms the state-of-the-art RW-based method in terms of accuracy and robustness to label noise, and that SybilSCAR outperforms the state-of-the-art LBP-based method in terms of accuracy, scalability, and convergence.
\end{packeditemize}

\section{Related Work}
\label{relatedwork}

\subsection{Structure-based Methods}
We classify structure-based methods into Random Walk (RW)-based methods and Loopy Belief Propagation (LBP)-based methods.
Structure-based methods aim to leverage social structure~\cite{ Yu06,Yu08,Danezis09,Mohaisen11,sybilrank,Yang12-spam,integro,zhang2016truetop,smartwalk,SybilWalk,sybilbelief,sybilframe,robustspammer,gang}.
The key intuition is that, although an attacker can control the connections between Sybils arbitrarily, it is harder for the attacker to manipulate the connections between benign nodes and Sybils, because such manipulation requires actions from benign nodes. 
Therefore, benign nodes and Sybils have a structural gap, which is leveraged by RW-based and LBP-based methods. 

\myparatight{RW-based methods} 
Example RW-based based methods include SybilGuard~\cite{Yu06}, SybilLimit~\cite{Yu08}, SybilInfer~\cite{Danezis09}, SybilRank~\cite{sybilrank}, Criminal account Inference Algorithm (CIA)~\cite{Yang12-spam},  \'{I}ntegro~\cite{integro}, and SybilWalk~\cite{SybilWalk}. 
Specifically, SybilGuard~\cite{Yu06} and SybilLimit~\cite{Yu08} assume that it is easy for short random walks starting from a labeled benign user to quickly reach other benign users, while hard for short random walks starting from Sybils to reach benign users.  SybilGuard and SybilLimit use the same RW lengths for all nodes. 
SmartWalk~\cite{smartwalk} leverages machine learning classifiers to predict the appropriate RW length for different nodes, and can improve the performance of SybilLimit via using the predicted (different) RW length for each node. 
SybilInfer~\cite{Danezis09} combines RWs with Bayesian inference and Monte-Carlo sampling  to directly detect the bottleneck cut between benign users and Sybils. 
SybilRank~\cite{sybilrank} uses short RWs to distribute benignness scores from a set of labeled benign users to all the remaining users. CIA~\cite{Yang12-spam} distributes badness scores from a set of labeled Sybils to other users. With a certain probability, CIA restarts the RW from the initial probability distribution, which is assigned based on the set of labeled Sybils.  \'{I}ntegro~\cite{integro} improves SybilRank by first leveraging victim  prediction (a victim is a user that connects to at least one Sybil) to assign weights to edges of a social network and then performing random walks on the weighted social network. 


Existing RW-based methods suffer from one or two key limitations: 1) they can only leverage either labeled benign users or labeled Sybils, but not both, which limits their detection accuracies; and 2) they are not robust to label noise in the training dataset.
Specifically, 
SybilGuard, SybilLimit, SybilInfer, and SmartWalk only leverage one labeled benign node, making their accuracy limited~\cite{sybilrank} and making them sensitive to label noise. Moreover, they are not scalable to large-scale social networks because they need to simulate a large number of random walks. SybilRank was shown to outperform a variety of Sybil detection methods~\cite{sybilrank}, and we treat it as a state-of-the-art RW-based method. 
SybilRank can only leverage the labeled benign users in a training dataset, which limits its detection accuracy, as we will demonstrate in our experiments.  Moreover, SybilRank is not robust to label noise, as we will demonstrate in our experiments. 
\myparatight{LBP-based methods}
LBP-based methods~\cite{sybilbelief,sybilframe,robustspammer,gang} also leverage the structure of the social network.
SybilBelief models a social network as a pairwise Markov Random Field (pMRF). 
Given some labeled Sybils and labeled benign users, SybilBelief first assigns prior probabilities to them and then uses LBP~\cite{Pearl88} to iteratively estimate the posterior probability of being a Sybil for each   remaining user.  The posterior probability of being a Sybil is used to predict a user's label.   SybilBelief can leverage both labeled Sybils and labeled benign users simultaneously, and it is robust to label noise~\cite{sybilbelief}. Gao et al.~\cite{sybilframe} and Fu et al.~\cite{robustspammer} demonstrated that SybilBelief can achieve better performance when learning the node and edge priors using local graph structure analysis. 
However, SybilBelief and its variants suffer from three limitations: 
1) they are not guaranteed to converge because LBP might oscillate on graphs with loops~\cite{Pearl88}; 
2) they are not scalable because LBP requires storing and maintaining messages on each edge; and 3) they do not have theoretically guaranteed performance. 
The first limitation means that their performance heavily relies on the number of iterations that LBP runs, but the best number of iterations might be different for different social networks. We note that Wang et al.~\cite{gang} recently proposed GANG, which generalized SybilBelief to directed social graphs (e.g., Twitter) and extended the techniques proposed in this work to make GANG scalable and convergent. 

\subsection{Other methods}
Some methods detect Sybils via binary machine learning classifiers. In particular, most methods in this direction represent each user using a set of features, which can be extracted from users' local subgraph structure (e.g., ego-network)~\cite{Yang11-sybil, Wang10} and side information (e.g., IP address, behaviors, and content)~\cite{yardi10,LeeUncovering10,benevenuto2010detecting,Song11,facebookImmune,Wang13Clickstream,CaoCCS14,cao2015combating}. Then, given a training dataset consisting of labeled benign users and labeled Sybils, they learn a binary classifier, e.g., logistic regression. Finally, the classifier is used to predict labels for the remaining users. A fundamental limitation of these methods is that attackers can mimic benign users by manipulating their profiles, so as to bypass the detection. However, we believe these methods can still be used to filter the basic Sybils.  
Moreover, these feature-based methods can be further combined with structure-based methods. For instance, for each user, the classifier can produce a probability that the user is a Sybil; such probabilities can be used as prior probabilities in LBP-based methods, e.g., SybilBelief~\cite{sybilbelief}. Indeed,   Gao et al.~\cite{sybilframe} generalized SybilBelief to incorporate such feature-based priors and demonstrated performance improvement.

\section{Problem Definition}
\label{probdef}

We formally define our structure-based Sybil detection problem, introduce our design goals, and describe the threat model we consider in the paper.

\subsection{Structure-based Sybil Detection}

Suppose we are given an undirected social network $G=(V,E)$, where a node $v \in V$ represents a user and an edge $(u,v) \in E$ indicates a mutual relationship between $u$ and $v$. $|V|$ and $|E|$ are number of nodes and edges, respectively.
For instance, on Facebook, an edge $(u,v)$ could mean that $u$ is in $v$'s friend list and vice versa. 
On Twitter, an edge $(u,v)$ could mean that $u$ follows $v$. Our structure-based Sybil detection is defined as follows:

\begin{definition}[Structure-based Sybil Detection] 
Suppose we are given a social network and a training dataset consisting of some labeled Sybils and labeled benign nodes. Structure-based Sybil detection is to predict the label of each remaining node by leveraging the global structure of the social network. 
\end{definition}

Like most existing studies on structure-based Sybil detection, we focus on undirected social networks. However, our methods can be generalized to directed social networks. For instance, Wang et al.~\cite{gang} generalized our methods to design a LBP-based method for directed social networks.

\subsection{Design Goals}
We target a method that satisfies the following goals:

\myparatight{1) Leveraging both labeled benign users and Sybils} Social network service providers often have a set of labeled benign users and labeled Sybils. For instance, verified users on Twitter or Facebook can be treated as labeled benign users; users spreading spam or malware can be treated as labeled Sybils, which can be obtained through manual inspection~\cite{sybilrank} or crowdsourcing~\cite{Wang13}. Our method should be able to leverage both labeled benign users and labeled Sybils to enhance  detection accuracy. 

\myparatight{2) Robust to label noise} A given label of a user is noisy if it does not match the user's true label. Labeled users may have noisy labels. For instance, an adversary could compromise a labeled benign user or 
make a  Sybil whitelisted as a benign user. In addition, labels obtained through manual inspection, especially crowdsourcing, often contain noises due to human mistakes~\cite{Wang13}.  We target a method that is robust when a minority fraction of given labels are incorrect. 

\myparatight{3) Scalable} Real-world social networks often have hundreds of millions of users and edges. Therefore, our method should be scalable and easily parallelizable. 

\myparatight{4) Convergent} Existing methods and our method are iterative methods. 
Convergence makes it easy to determine when to stop an iterative method. 
 It is hard to set the best number of iterations for an iterative method that is not convergent.  
 Therefore, our method should be convergent.

\myparatight{5) Theoretical guarantee} Our method should have a theoretical guarantee on the number of Sybils that can be falsely accepted into a social network. This theoretical guarantee is important for security-critical applications that leverage social networks, e.g., social network based Sybil defense in peer-to-peer and distributed systems~\cite{Yu06}, and social network based anonymous communications~\cite{DanezisAnonyComPET10}. 

Existing RW-based SybilGuard~\cite{Yu06} and SybilLimit~\cite{Yu08} do not satisfy requirements 1), 2), and 3); SybilInfer~\cite{Danezis09} only satisfies the requirement 4); SybilRank~\cite{sybilrank}  and \'{I}ntegro~\cite{integro} do not satisfy requirements 1) and 2); CIA~\cite{Yang12-spam} does not satisfy requirements 1), 2), and 5). Existing LBP-based SybilBelief~\cite{sybilbelief} and SybilFuse~\cite{sybilframe} do not satisfy requirements 3), 4), and 5).

\subsection{Threat Model}
\label{threatmodel}

\begin{figure}[t]
\center
{\includegraphics[width=0.42\textwidth]{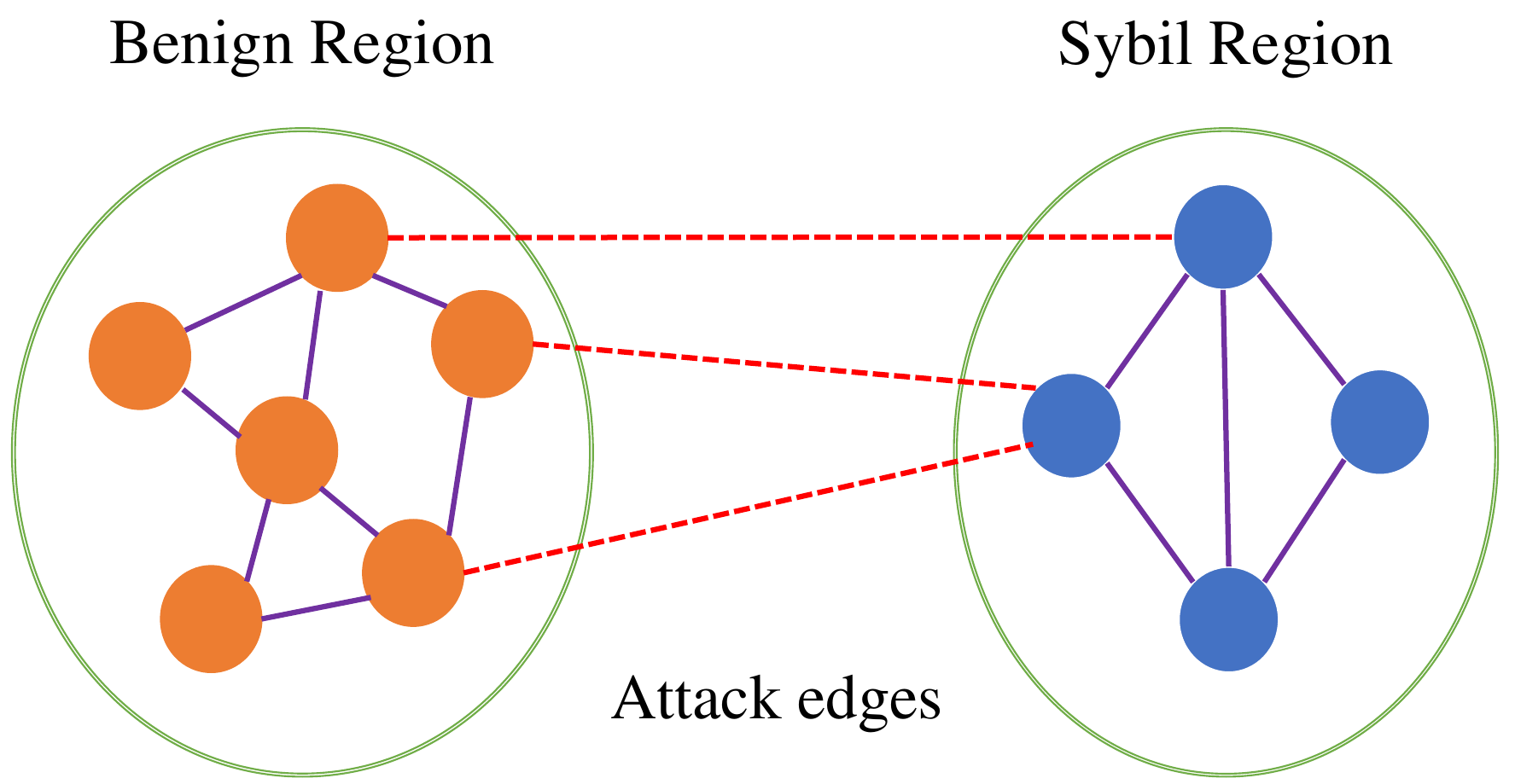}}
\caption{Benign region, Sybil region, and attack edges.}
\label{graph_attacks}
\end{figure}

We call the subgraph containing all benign nodes and edges between them the \emph{benign region}, and call the subgraph containing all Sybil nodes and edges between them the \emph{Sybil region}. Edges between the two regions are called \emph{attack edges}. Figure~\ref{graph_attacks} illustrates these concepts.

One basic assumption under structure-based Sybil detection methods is that the benign region and the Sybil region are sparsely connected (i.e., the number of attack edges is relatively small), compared with the edges among the two regions. In other words,  most benign users would not establish trust relationships with Sybils. We note that this assumption is equivalent to requiring that the social network follows \emph{homophily}, i.e., two linked nodes share the same label with a high probability. For an extreme example, if the benign region and the Sybil region are separated from each other, then the social network has a perfect homophily, i.e., every two linked nodes have the same label. 
Note that, it is of great importance to obtain social networks that satisfy this assumption, otherwise the detection accuracies of structure-based methods are limited.
For instance, Yang et al.~\cite{yang2014uncovering} showed that RenRen \emph{friendship} social network does not satisfy this assumption, and thus the performance of structure-based methods are unsatisfactory.
However, Cao et al.~\cite{sybilrank} found that Tuenti, the largest online social network in Spain,  satisfies the homophily assumption, and thus SybilRank can detect a large amount of Sybils in Tuenti. 

Generally speaking, there are two ways for  service providers to construct a social network that satisfies homophily. One way is to approximately obtain trust relationships between users by looking into user interactions~\cite{wilson:eurosys09}, predicting tie strength~\cite{gilbert:chi09}, asking users to rate their social contacts~\cite{sybildefender}, etc. The other way is to preprocess the network structure so that structure-based methods are suitable to be applied. Specifically, analysts could 
 detect and remove compromised benign nodes (e.g., front peers)~\cite{wang2010poisonedwater}, or employ feature-based classifier to filter Sybils, so as to decrease the number of attack edges and enhance the homophily. For instance, Alvisi et al.~\cite{alvisiSybil13} showed that if the attack edges are established randomly,  simple feature-based classifiers are sufficient to enforce Sybils to be suitable for structure-based Sybil detection. 
We note that the reason why the RenRen {friendship} social network did not satisfy homophily in the study of Yang et al. is that RenRen even didn't deploy simple feature-based classifiers at that time~\cite{yang2014uncovering}. 

Formally, we measure homophily as the fraction of edges in the social network that are not attack edges. For the same benign region and Sybil region, more attack edges indicate weaker homophily. 
As we will demonstrate in our empirical evaluations, our SybilSCAR can tolerate weaker homophily than existing methods. 

When analyzing the theoretical bound on the falsely accepted Sybils, SybilSCAR further assumes that the iterative process of SybilSCAR converges fast in the benign region, which is similar to the fast mixing assumption of RW-based methods.

\section{Our Local Rule-based Framework}

In this section, we unify existing RW-based methods~\cite{Mohaisen11,sybilrank,Yang12-spam,integro, smartwalk} and LBP-based methods~\cite{sybilbelief, sybilframe} into a local rule-based framework. 
Specifically, these methods first assign the \emph{prior knowledge} of all nodes using a training dataset. Then, they propagate the prior knowledge among the social network to obtain the \emph{posterior knowledge} via iteratively applying their \emph{local rules} to every node. \emph{A local rule is to update the posterior knowledge of a node by combining the influences from its neighbors with its prior knowledge}. We call the influence from a neighbor \emph{neighbor influence}. Figure~\ref{unifiedview} shows our unified framework.

\begin{figure}[t]
\center
{\includegraphics[width=0.4\textwidth]{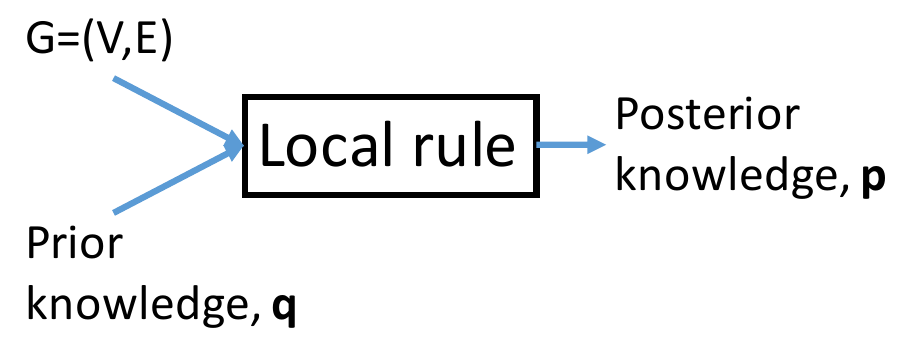}}
\caption{Our proposed framework to unify state-of-the-art RW-based and LBP-based Sybil detection methods.}
\label{unifiedview}
\end{figure}

\myparatight{Notations} We denote by $w_{uv}$ the weight of the edge $(u,v)$, $\Gamma_u$ the set of neighbors of node $u$, and $d_{u}$ the total weights of edges linked to $u$, i.e., $d_{u}=\sum_{v\in \Gamma_u} w_{uv}$.  In RW-based methods~\cite{Mohaisen11,Yang12-spam,integro}, edge weights model the relative importance (e.g., level of trust) of edges. In LBP-based method~\cite{sybilbelief},  an edge weight $w_{uv}$ models the tendency that  $u$ and $v$ share the same label. 
We denote by $q_u$ and $p_u$ the {prior knowledge} and {posterior knowledge} of the node $u$, respectively. In RW-based methods, $q_u$ and $p_u$ are the prior and posterior reputation scores of $u$, respectively, and they represent relative benignness of nodes. In LBP-based methods, $q_u$ and $p_u$ are the prior and posterior probabilities that node $u$ is a Sybil, respectively.

\subsection{Additive Local Rule of RW-based Methods}
State-of-the-art RW-based methods~\cite{Mohaisen11,sybilrank,Yang12-spam, integro} first assign prior reputation scores for every node using a training dataset. Then they iteratively apply
 the following local rule to every node:   
\begin{align}
\label{rwrule}
p_u = (1-\alpha)\sum_{v\in\Gamma_u} \underbrace{p_v\frac{w_{uv}}{d_{v}}}_{\mathrm{neighbor~ influence}} + \alpha \underbrace{q_u}_{\mathrm{prior~ knowledge}},
\end{align}
where $\alpha\in [0,1]$ is called a restart probability of the random walk. 
We note that SybilRank uses a restart probability of 0 and normalizes the final reputation scores by node degrees. 

We have two observations for the additive local rule. First, the neighbor influence from a neighbor $v$ to $u$ is a fraction of $v$'s current reputation score $p_v$, and the fraction is proportional to the edge weight $w_{uv}$.  Second, this local rule combines the prior knowledge and the neighbor influences \emph{linearly} to update the posterior knowledge about a node.

\subsection{Multiplicative Local Rule of LBP-based Methods}
SybilBelief~\cite{sybilbelief}, a LBP-based method, associates a binary random variable $x_u$ with each node $u$, where $x_u=1$ indicates that $u$ is Sybil while $x_u=-1$ indicates that $u$ is benign. Then, $q_u$ and  $p_u$ are the prior and posterior probabilities that $x_u=1$, respectively. SybilBelief first assigns the prior probabilities for nodes using
 a set of labeled benign nodes and/or a set of labeled Sybils, and then it iteratively applies the following local rule~\cite{sybilbelief}:
\begin{align}
\label{lbpneighbor}
&\underbrace{m_{vu}(x_u)}_{\mathrm{neighbor~ influence}} = \sum_{x_v} \phi_v(x_v) \varphi_{vu}(x_v, x_u) \prod_{z \in \Gamma_v/u} m_{zv}(x_v) \\
\label{lbpcombine}
&p_u = \frac{\overbrace{q_u}^{\mathrm{prior~ knowledge}} \prod_{v \in \Gamma_u} m_{vu}(1)}{q_u \prod_{v \in \Gamma_u} m_{vu}(1) + (1-q_u) \prod_{v \in \Gamma_u} m_{vu}(-1)},
\end{align}
where node potential $\phi_v(x_v)$ and edge potential $\varphi_{vu}(x_v, x_u)$ are defined as follows:
\begin{displaymath}
\phi_v(x_v):=
\begin{cases} 
q_v  & \text{if } x_v = 1  \\ 
1- q_v  &  \text{if } x_v = -1
\end{cases} 
\end{displaymath}
\begin{displaymath}
\varphi_{vu}(x_v, x_u):=
\begin{cases} 
w_{vu} & \text{if } x_u x_v = 1  \\  
1 - w_{vu} &  \text{if } x_u x_v = -1,
\end{cases}
\end{displaymath}

We also have two observations for the multiplicative local rule. First, \emph{this local rule explicitly models neighbor influences}. Specifically, the neighbor influence from a neighbor $v$ to $u$ (i.e., $m_{vu}(x_u)$) is defined in Equation~\ref{lbpneighbor}. To compute the neighbor influence $m_{vu}(x_u)$, $u$'s neighbor $v$ needs to multiply the neighbor influences from all its neighbors except $u$.
Second, according to Equation~\ref{lbpcombine}, this local rule combines the neighbor influences with the prior probability \emph{nonlinearly}.  

\subsection{Comparing RW-based Additive Local Rule with LBP-based Multiplicative Local Rule} 

LBP-based multiplicative local rule can tolerate a relatively larger fraction of label noise because of its nonlinearity~\cite{sybilbelief}, and it can leverage both labeled benign nodes and labeled Sybils. However,  LBP-based multiplicative local rule is space and time inefficient because it requires a large amount of space and time to maintain the neighbor influences associated with every edge, and methods using this local rule are not guaranteed to converge.
In contrast, RW-based additive local rule is space and time efficient, and methods using this local rule are guaranteed to converge.  However, this local rule is sensitive to label noise, and it cannot leverage labeled benign nodes and labeled Sybils simultaneously.

\section{Design of SybilSCAR}

Under our framework, designing a new Sybil detection method is reduced to designing a new local rule. Therefore, 
we design a novel local rule. Then, we describe how we design SybilSCAR based on the new local rule.

\subsection{Our New Local Rule}

We aim to design a local rule that  integrates the advantages of both RW-based and LBP-based local rules, while overcoming their limitations. Roughly speaking, our idea is to leverage the \emph{multiplicativeness} like LBP-based local rule  to be robust to label noise, while \emph{avoiding maintaining neighbor influences} to be as space and time efficient as RW-based local rule. Next, we show how we model neighbor influences and combine neighbor influences with prior knowledge.

\myparatight{Neighbor influence} 
We associate a binary random variable $x_u$ with a node $u$, where $x_u=1$ and  $x_u=-1$ mean that $u$ is a Sybil and benign node, respectively. 
We denote $p_u$  as the posterior probability that $u$ is a Sybil, i.e., $p_u = Pr(x_u=1)$. $p_u > 0.5$ means $u$ is more likely to be a Sybil; $p_u < 0.5$ means $u$ is more likely to be benign; and $p_u = 0.5$ means we cannot decide $u$'s label. 
We model the probability that $u$ and $v$ have the same label as $w_{vu} \in [0, 1]$, which defines the \emph{homophily strength} of the edge $(u,v)$. 
Formally, we have:
\begin{align}
\text{Pr}(x_u=1|x_v=1)=\text{Pr}(x_u=-1|x_v=-1)=w_{vu}, \nonumber \\  
\text{Pr}(x_u=1|x_v=-1)=\text{Pr}(x_u=-1|x_v=1)=1-w_{vu}. \nonumber
\end{align}
$w_{uv}>0.5$ means that $u$ and $v$ are in a \emph{homogeneous relationship}, i.e., they tend to share the same label;  $w_{uv}<0.5$ means that $u$ and $v$ are in a \emph{heterogeneous relationship}, i.e., they tend to have the opposite labels; and $w_{uv}=0.5$ means that $u$ and $v$ are not correlated.

We denote by $f_{vu}$ the \emph{neighbor influence} of a neighbor $v$ to $u$.  
$f_{vu}$ is defined as the probability that $u$ is a Sybil (i.e., $x_u=1$), given the neighbor $v$'s information alone.
According to the \emph{law of total probability}, we compute $f_{vu}$ as: 
\begin{align}
\label{ourneighbor}
f_{vu} &=	 \text{Pr}(x_u=1| x_v=1)\text{Pr}(x_v=1) \nonumber\\
	&+ \text{Pr}(x_u=1| x_v=-1)\text{Pr}(x_v=-1) \nonumber \\
	&=w_{vu}p_v + (1-w_{vu})(1-p_v).
\end{align}
We have several observations from Equation~\ref{ourneighbor}:
\begin{packeditemize}
\item $v$ has no neighbor influence to $u$ (i.e., $f_{vu}=0.5$) if $v$'s label is undecidable (i.e., $p_v=0.5$) or $u$ and $v$ are uncorrelated (i.e., $w_{vu} = 0.5$); 
\item $v$ has a positive neighbor influence to $u$ if $v$ and $u$ are in a homogeneous relationship, i.e., if $ p_v>0.5$ $(or <0.5) \text{ and }  w_{vu}>0.5$, then  $f_{vu}>0.5$ $( or <0.5)$; 
\item $v$ has a negative neighbor influence to $u$ if $v$ and $u$ are in a heterogeneous relationship, i.e., if $ p_v>0.5$  $(or <0.5) \text{ and }  w_{vu}<0.5$, then  $f_{vu}<0.5$  $(or >0.5)$. 
\end{packeditemize}

\myparatight{Combining neighbor influences with prior} In our local rule, a node's posterior probability of being Sybil is updated by combining its neighbor influences with its prior probability of being Sybil. In order to tolerate label noise, we leverage the multiplicative local rule in LBP-based methods. Specifically, we have:
\begin{align}
\label{ourcombine}
p_u = \frac{q_u \prod_{v \in \Gamma_u} f_{vu}}{q_u \prod_{v \in \Gamma_u} f_{vu} + (1-q_u) \prod_{v \in \Gamma_u} (1-f_{vu})}.
\end{align}

However, methods that iteratively apply the above multiplicative local rule to every user are not guaranteed to converge. Therefore we further \emph{linearize} Equation~\ref{ourcombine}. 
We first define two concepts \emph{residual variable} and \emph{residual vector}.

\begin{definition}[Residual Variable and Vector] 
\label{def_1}
We define the residual of a variable $y$ as $\hat{y}=y-0.5$; and we define the residual vector $\hat{\mathbf{y}}$ of $\mathbf{y}$ as $\hat{\mathbf{y}} = [y_1-0.5, y_2-0.5, \cdots]$.
\end{definition}

With above definition, we denote $\hat{w}_{vu}$ as the residual homophily strength.  
Moreover, by substituting variables in Equation~\ref{ourneighbor} with their corresponding residuals, we have 
the residual  neighbor influence $\hat{f}_{vu}$ as follows:
\begin{align}
\label{res_nei_influ_sybil}
\hat{f}_{vu} = 2 \hat{p}_{v} \hat{w}_{vu}.
\end{align}

Based on the approximations $\ln (1+x) \approx x$ and  $\ln (1- x) \approx -x$ when $x$ is small, we have the following theorem, which linearizes Equation~\ref{ourcombine}.

\begin{theorem}
\label{theorem_1}
The residual posterior probability of being a Sybil for a node $u$ can be linearized as:
\begin{equation}
\label{res_bel}
\hat{p}_{u}  = \hat{q}_{u} + \sum_{v \in \Gamma(u)} \hat{f}_{vu}.
\end{equation}
\end{theorem}

\begin{proof}
See Appendix~\ref{app:theorem1}.
\end{proof}

By combining Equation~\ref{res_nei_influ_sybil} and Equation~\ref{res_bel}, we obtain our new local rule as follows:
\begin{align}
\label{ourlocalrule}
\text{\bf Our local rule:} \quad 
\hat{p}_{u}  = \hat{q}_{u} + 2 \sum_{v \in \Gamma(u)} \hat{p}_{v} \hat{w}_{vu}.
\end{align}

\subsection{SybilSCAR Algorithm} 
Our SybilSCAR iteratively applies our local rule to every node to compute the posterior probabilities. 
Suppose we are given a set of labeled Sybils which we denote as $L_s$ and a set of labeled benign nodes which we denote as $L_b$. SybilSCAR first utilizes $L_s$ and $L_b$ to assign a prior probability of being a Sybil for all nodes. Specifically, 
 \begin{align}
 \label{prior}
 q_u=
 \begin{cases} 
0.5 + \theta  & \text{if } u \in L_s   \\ 
0.5 - \theta  &  \text{if } u \in L_b \\
0.5 & \text{otherwise},
\end{cases} 
 \end{align}
where $\theta >0$ indicates that we assign a higher prior probability of being a Sybil to labeled Sybils. Considering that the labels might have noise, we will set $\theta$ to be smaller than 0.5. 
In practice, these prior probabilities can also be obtained from feature-based methods. Specifically, for each user we can leverage a binary classifier, trained using user's local features, to produce the probability of being a Sybil, which  can then be treated as the user's prior probability. With such prior probabilities, SybilSCAR iteratively applies our local rule in Equation~\ref{ourlocalrule} to update residual posterior probabilities of all nodes.

\myparatight{Representing SybilSCAR as a matrix form} For convenience, we denote by a vector ${\bf q}$ the prior probability of being a Sybil for all nodes, i.e., ${\bf q}=[q_1; q_2; \cdots ; q_{|V|}]$. Similarly, we denote by a vector ${\bf p}$  the posterior probability of all nodes, i.e., ${\bf p}=[p_1; p_2; \cdots; p_{|V|}]$. Moreover, we denote $\hat{\mathbf{q}}$ and $\hat{\mathbf{p}}$ as the residual prior probability vector and residual posterior probability vector of all nodes, respectively. 
We denote $\mathbf{A}\in \mathbb{R}^{|V| \times |V|}$ as the adjacency matrix of the social graph, where the $u$th row represents the neighbors of $u$. 
Formally, if there exists an edge $(u,v)$ between nodes $u$ and $v$, then the entry $A_{uv}=A_{vu}=1$, otherwise  $A_{uv}=A_{vu}=0$. 
{Moreover, we denote $\hat{\mathbf{W}}$ as the corresponding residual homophily strength matrix , where $\hat{w}_{vu} = 0$ if $A_{vu}=0$.}
With these notations, we can represent our SybilSCAR as iteratively applying the following equation:
\begin{align}
\label{approx_update}
 \hat{\mathbf{p}}^{(t)}= \hat{\mathbf{q}} + 2 \hat{\mathbf{W}} \hat{\mathbf{p}}^{(t-1)},
\end{align}
where $ \hat{\mathbf{p}}^{(t)}$ is the residual posterior probability vector in the $t$th iteration. 
Initially, we set $ \hat{\mathbf{p}}^{(0)}=\hat{\mathbf{q}}$.

{Algorithm}~\ref{alg:SybilSCAR} summarizes the  pseudocode of SybilSCAR. 
We stop running SybilSCAR when the relative errors of residual posterior probabilities between two consecutive iterations is smaller than some threshold $\delta$ or it reaches the predefined number of maximum iterations $T$. 
After SybilSCAR halts, we predict $u$ to be a Sybil if $p_u > 0.5$, otherwise we predict $u$ to be benign. 

In this paper, we consider the following two cases for the homophily strength $\hat{\mathbf{W}}$. Although we study these two settings, we believe that learning the homophily strength for each edge would be a valuable future work.

\myparatight{SybilSCAR-C} In this variant, we use a constant homophily strength for all edges, i.e., $\hat{w}_{vu} = \hat{w}$. 

\myparatight{SybilSCAR-D} In this variant, we use a degree-normalized homophily strength for each edge, i.e., $\hat{w}_{vu} = \frac{1}{2 d_u}$, where $d_u$ is the degree of node $u$. 
The intuition is that when a node has many neighbors, each neighbor has a small influence on the node. In this variant, a node's residual posterior probability is the sum of its residual prior probability and the average residual posterior probability of its neighbors. We note that SybilSCAR-D essentially uses the RW-based \emph{neighbor influence} proposed by SybilWalk~\cite{SybilWalk}. The differences with SybilWalk include 1) SybilWalk resets the residual posterior probabilities of labeled nodes to their residual prior probabilities in each iteration, and 2) SybilWalk does not consider prior probabilities of unlabeled nodes.

\begin{algorithm}[!t]
\caption{SybilSCAR}
\label{alg:SybilSCAR}
\begin{algorithmic}
\REQUIRE $G=(V,E)$, $L_s$, $L_b$, $\theta$, $\hat{\mathbf{W}}$, $\delta$, and $T$. \\
\ENSURE  $p_u, \forall \, u \in V$. \\
    Initialize $\hat{\mathbf{p}}^{(0)} = \hat{\mathbf{q}}$. \\
    Initialize $t=1$. \\
    \WHILE { $\frac{\|\hat{\mathbf{p}}^{(t)} - \hat{\mathbf{p}}^{(t-1)}\|_1}{\|\hat{\mathbf{p}}^{(t)}\|_1} \geq \delta$ and $t \leq T$} 
    \STATE Update residual posterior vector $\hat{\mathbf{p}}^{(t)}$ using Equation~\ref{approx_update}. 
    \STATE  $t=t+1$. 
    \ENDWHILE \\
    \RETURN $\hat{\mathbf{p}}^{(t)} + 0.5$. \\
\end{algorithmic}
\end{algorithm}

\section{Theoretical Analysis}
We first analyze the convergence condition of SybilSCAR.
Then we analyze its peformance bound. 
Finally, we analyze its computational complexity.

\subsection{Convergence Condition}
\label{sec:cvgAnal}
We analyze the condition when SybilSCAR converges.

\begin{lemma}[Sufficient and Necessary Convergence Condition for a Linear System~\cite{saad2003iterative}] 
\label{linearsystem}
Suppose we are given an iterative linear process: $\mathbf{y}^{(t)} \leftarrow \mathbf{c} + \mathbf{M} \mathbf{y}^{(t-1)}$. 
The linear process converges with any initial choice $\mathbf{y}^{(0)}$ if and only if the spectral radius\footnote{The spectral radius of a square matrix is the maximum of the absolute values of its eigenvalues.} of $\mathbf{M}$ is smaller than 1, i.e., $\rho(\mathbf{M}) < 1$. 
\end{lemma}
\begin{proof}
See~\cite{saad2003iterative}.
\end{proof}

Based on Equation~\ref{approx_update} and {Lemma}~\ref{linearsystem}, we are able to analyze the convergence condition of SybilSCAR.  

\begin{theorem}[Sufficient and Necessary Convergence Condition of SybilSCAR] 
\label{lemma_2}
The sufficient and necessary condition that makes SybilSCAR converge is equivalent to 
\begin{align}
\label{suff_ness_cond}
\begin{split}
 \rho(\hat{\mathbf{W}}) < \frac{1}{2}.
\end{split}
\end{align}

\end{theorem}
\begin{proof}
{By directly using {Lemma}~\ref{linearsystem} in Equation~\ref{approx_update}.}
\end{proof}

{Theorem}~\ref{lemma_2} provides a strong sufficient and necessary convergence condition.
However, in practice 
using Theorem~\ref{lemma_2} is computationally expensive, as it involves computing the largest eigenvalue with respect to spectral radius of $\hat{\mathbf{W}}$. Hence, we instead derive a \emph{sufficient condition} for SybilSCAR's convergence, which enables us to set $\hat{w}$  with cheap computation. Specifically, our sufficient condition is based on the fact that any norm is an upper bound of the spectral radius~\cite{derzko1965bounds}, i.e., $\rho(\mathbf{M}) \leq \| \mathbf{M} \|$, where $\| \cdot \|$ indicates some matrix norm. 
In particular, we use the induced $l_\infty$ matrix norm $\| \cdot \|_\infty$ \footnote{$\| \mathbf{M} \|_\infty = \max_{i} \sum_{j} | \mathbf{M}_{ij} | $, the maximum absolute row sum of the matrix.}. 
In this way, our sufficient condition for convergence is as follows:

\begin{theorem}[Sufficient Convergence Condition of SybilSCAR] 
\label{lemma_3}
A sufficient condition that makes SybilSCAR converge is
\begin{align}
\label{suff_cond}
\| \hat{\mathbf{W}} \|_\infty < \frac{1}{2}.  
\end{align}
\end{theorem}
\begin{proof}
{As $\rho(\hat{\mathbf{W}}) \leq \|\hat{\mathbf{W}}\|_\infty $, we achieve the sufficient condition by enforcing $2 \|\hat{\mathbf{W}}\|_\infty < 1$, and thus $\|\hat{\mathbf{W}}\|_\infty < \frac{1}{2}$.}
\end{proof}

{Next, we derive the sufficient convergence condition for SybilSCAR-C and SybilSCAR-D, respectively.} 

{
\myparatight{SybilSCAR-C} 
By applying Theorem~\ref{lemma_3}, 
we have the sufficient condition for SybilSCAR-C to converge as  
$\hat{w} < \frac{1}{2 \|\mathbf{A}\|_\infty} = \frac{1}{2 \max_{u\in V} d_u}$.
Note that our result provides a guideline to set $\hat{w}$, i.e., once $\hat{w}$ is smaller than the inverse of 2 times of the maximum node degree, SybilSCAR-C is guaranteed to converge. 
In practice, however, some nodes (e.g., celebrities) could have orders of magnitude bigger degrees than the others (e.g., ordinary people), and such nodes make $\hat{w}$ very small. 
In our experiments, we found that SybilSCAR-C can still converge when replacing the maximum node degree with the average node degree.

\myparatight{SybilSCAR-D} 
In this case, the summation of each row of $\hat{\mathbf{W}}$ has a fixed value $\frac{1}{2} $. Therefore, $ \|\mathbf{W}\|_\infty  = \frac{1}{2} $. 
In practice, it is often $\rho(\hat{\mathbf{W}}) < \|\hat{\mathbf{W}}\|_\infty$, which implies that $\rho(\hat{\mathbf{W}}) < \frac{1}{2}$. Therefore, SybilSCAR-D is also convergent. 
}

\subsection{Security Guarantee}
\myparatight{Existing RW-based methods}
Some existing RW-based Sybil detection methods~\cite{Yu06, Yu08, Danezis09, sybilrank, integro} have theoretical guarantees on the number of Sybils that are falsely accepted into a social network. For instance, Table~\ref{theoryDetectionAccuracy} shows the theoretical guarantees of some representative methods. 
These guarantees are achieved based on the assumption that the benign region of the social network is  \emph{fast-mixing}~\cite{levin2009markov}. 
Roughly speaking, a graph is fast mixing if a random walk on the graph converges to its stationary distribution in $O(\log |V|)$ iterations. 

\begin{table}[!t]\renewcommand{\arraystretch}{1.8}
\centering
\footnotesize
\caption{Summary of theoretical guarantees of various structure-based methods. $g$ is the number of attack edges (sum of weights on the attack edge for \'{I}ntegro) and $d(B)_{\min}$ is the minimum node degree in the benign region.  SybilGuard requires $g=o(\sqrt{|V|}/ \log |V|)$. The symbol ``--" means the corresponding bound is unknown.}
\centering
\begin{tabular}{|c|c|c|} \hline 
{\small Method} & {\small \#Accepted Sybils} \\ \hline
{\small SybilGuard~\cite{Yu06}} & {\small $O(g\sqrt{|V|} \log |V|)$}  \\ \hline
{\small SybilLimit~\cite{Yu08}} & {\small $O(g \log |V|)$}\\ \hline
{\small SybilInfer~\cite{Danezis09}} & {\small --} \\ \hline
{\small SybilRank~\cite{sybilrank}} & {\small $O(g\log |V|)$} \\ \hline
{\small CIA~\cite{Yang12-spam}} & {\small -- } \\ \hline
{\small \'{I}ntegro~\cite{integro}} & {\small $O(g\log |V|)$} \\ \hline
{\small SybilBelief~\cite{sybilbelief}} & {\small -- } \\ \hline
{\small SybilSCAR-D} & {\small  $O(\frac{g\log |V|}{d(\mathcal{S})})$} \\ \hline
\end{tabular}
\label{theoryDetectionAccuracy}
\end{table}

\myparatight{SybilSCAR} We will derive security guarantee for a ``weaker" version of SybilSCAR-D. 
In SybilSCAR-D, in each iteration, a node's residual posterior probability is the sum of its residual prior probability and the average residual posterior probability of its neighbors. In other words, in each iteration, a node's prior probability is injected to influence the dynamics of nodes' residual posterior probabilities, which makes it harder to analyze the dynamics of residual posterior probabilities. Therefore, we consider a weaker version of SybilSCAR-D, in which the nodes' residual prior probabilities are only injected in the initialization step. In other words, we have 
\begin{align}
\hat{\mathbf{p}}^{(0)} & =\hat{\mathbf{q}} \\
\hat{p}_{u}^{(t+1)}  &=  \sum_{v \in \Gamma(u)} \frac{\hat{p}_{v}^{(t)}}{d_u}.
\end{align}
This version of SybilSCAR-D has a converged solution that every node has the same residual posterior probability, i.e., $\hat{p}_{u}=\pi$ for every node $u$ is a solution for SybilSCAR-D.  
We have the following security guarantee for this version of SybilSCAR-D:

\begin{theorem}
\label{theorem:bound}
 Suppose SybilSCAR-D only leverages the nodes' prior probabilities in the initialization step, residual posterior probabilities in the benign region converge in $O(\log |V|)$) iterations (this is similar to the fast mixing assumption of RW-based methods), the attacker randomly establishes $g$ attack edges, and we are only given some  labeled benign nodes. Then, the total number of Sybils whose residual posterior probabilities of being Sybil are lower than those of certain benign nodes is bounded by  $O(\frac{g\log |V|}{d(\mathcal{S})})$,  
where $d(S)$ is the average node degree in the Sybil region. 
\end{theorem} 

\begin{proof}
See Appendix~\ref{app:analysis}.
\end{proof}

\begin{figure}[!t]
\centering
\subfigure[ER model]{\includegraphics[width=0.24\textwidth]{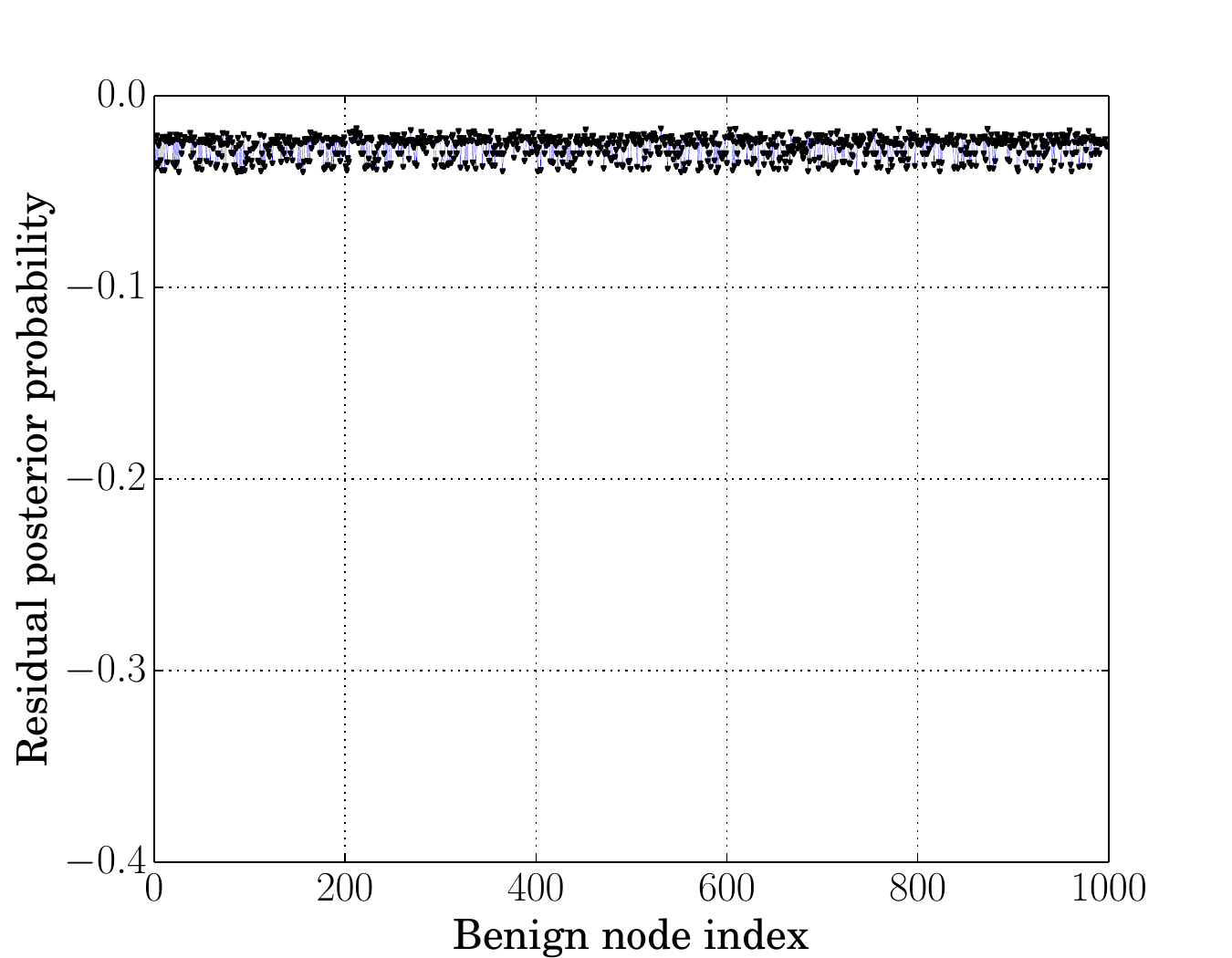}\label{ER}}
\subfigure[PA model]{\includegraphics[width=0.24\textwidth]{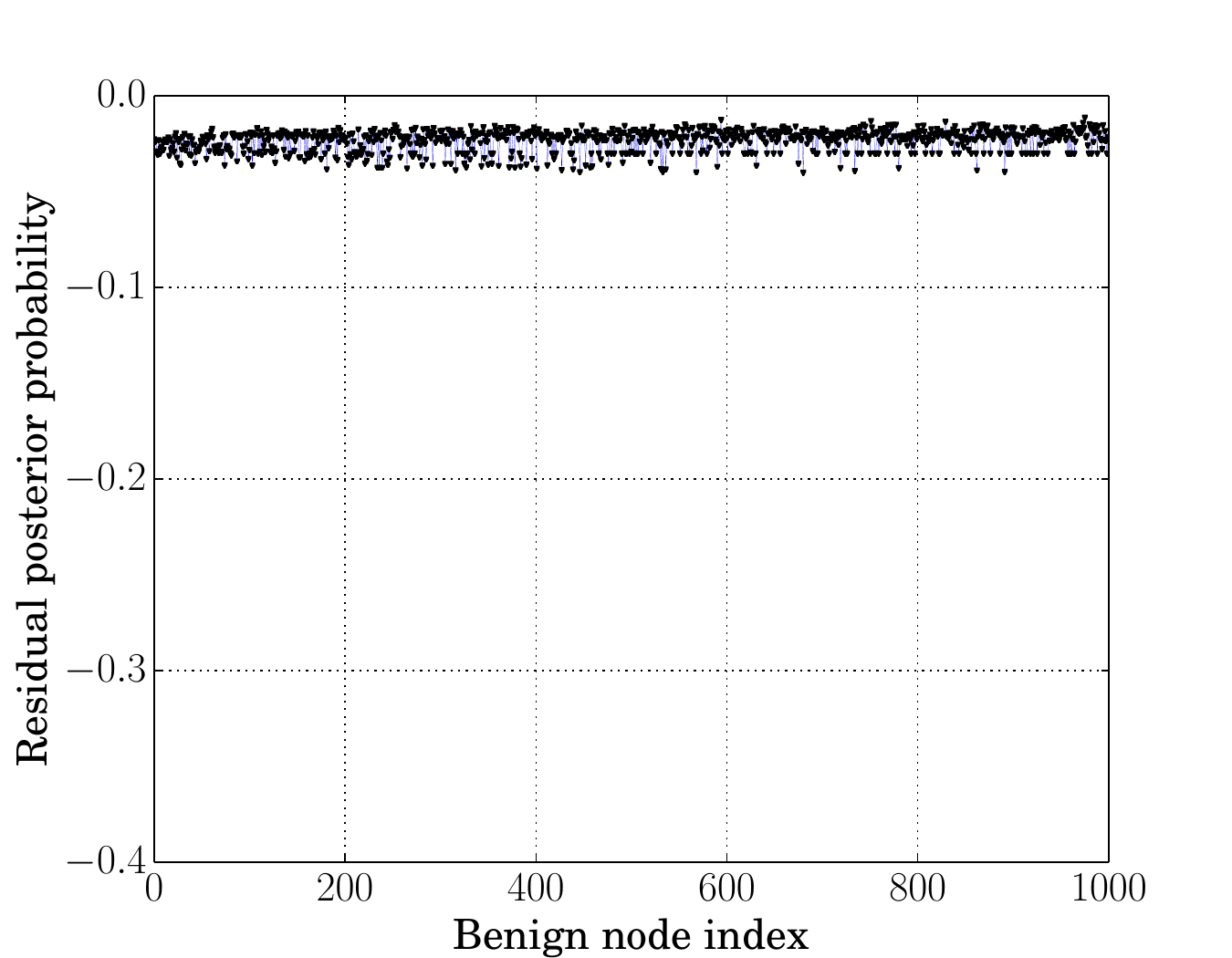}\label{PA}}
\caption{Residual posterior probabilities of unlabeled benign nodes.}
\label{ER-PA}
\end{figure}

Theorem~\ref{theorem:bound} implies that when Sybils are more densely connected among themselves (i.e., the average degree $d(S)$ is larger), it is easier for SybilSCAR to detect them. An intuitional explanation is that, when the Sybil region is more dense, a larger proportion of the residual posterior probabilities would be propagated among the Sybil region. 
Table~\ref{theoryDetectionAccuracy} summarizes the theoretical performance bound of existing structure-based methods. For SybilRank, \'{I}ntegro, and CIA, the metric \emph{\#accepted Sybils} means the number of Sybils that are ranked lower than certain benign nodes. For the rest of methods, \#accepted Sybils means the number of Sybils that are classified as benign. 
As we can see, our SybilSCAR achieves the tightest bound on the number of falsely accepted Sybils. We note that deriving security guarantee for SybilSCAR-C is still an open challenge. However, as we will demonstrate in our experiments, SybilSCAR-C outperforms SybilSCAR-D.

{One key assumption of Theorem~\ref{theorem:bound} is that residual posterior probabilities in the benign region converge after $O(\log |V|)$) iterations. In other words, nodes in the benign region have similar residual posterior probabilities after $O(\log |V|)$) iterations. We validate this assumption via simulations. Specifically, 
we synthesize a benign region and a Sybil region with 1,000 nodes and an average degree of 40 via the Erdos--Renyi (ER) model~\cite{erdos1960evolution} or the Preferential Attachment (PA) model~\cite{Barabasi99}; we randomly add 1,000 attack edges between the two regions; and we randomly label 10 benign nodes as the training set. 
Figure~\ref{ER-PA} shows the residual posterior probabilities of unlabeled benign nodes after $\log |V|$ iterations. We observe that unlabeled benign nodes have similar residual posterior probabilities. }

\subsection{Complexity Analysis}
\label{sec:complexity}

SybilSCAR (both SybilSCAR-C and SybilSCAR-D), state-of-the-art RW-based methods~\cite{Mohaisen11,sybilrank,Yang12-spam}, and LBP-based method~\cite{sybilbelief} have the same space complexity, i.e., $O(|E|)$, and their time complexity is $O(t|E|)$, where $t$ is the number of iterations. 
Although SybilSCAR and SybilBelief (a LBP-based method) have the same asymptotic space and time complexity, SybilSCAR
is several times more space efficient and significantly more time efficient than SybilBelief in practice, as we demonstrate in our experiments. This is because SybilBelief needs to store neighbor influences (i.e., $m_{vu}(x_u)$)  in both directions of every edge and update them in every iteration.
 
\myparatight{Parallel implementation} SybilSCAR, state-of-the-art RW-based methods~\cite{Mohaisen11,sybilrank,Yang12-spam}, and LBP-based methods~\cite{sybilbelief, sybilframe} can be easily implemented in parallel. Specifically, we can divide nodes into groups, and a thread or computer applies the corresponding local rule to a group of nodes iteratively.

\section{Empirical Evaluations}

{We compare our SybilSCAR-C and SybilSCAR-D with SybilRank~\cite{sybilrank}, a state-of-the-art RW-based method, and SybilBelief~\cite{sybilbelief}, a state-of-the-art LBP-based method, in terms of accuracy, robustness to label noise, scalability, and convergence.}

\subsection{Experimental Setup}
\label{exp-setup}

\begin{table}[!t]\renewcommand{\arraystretch}{1.2}
\centering
\caption{Dataset statistics.}
\label{dataset_stat}
\begin{tabular}{|c|c|c|c|} 
 \hline
 {\small \textbf{Dataset}} & {\small \#Nodes} & {\small \#Edges}  &  {\small Ave. degree} \\ \hline
 {\small \textbf{Facebook}} & {\small 4,039} & {\small 88,234} & {\small 43.69} \\ \hline
 {\small \textbf{Enron}}& {\small 33,696} & {\small 180,811} & {\small 10.73} \\ \hline
 {\small \textbf{Epinions}}& {\small 75,877} & {\small 811,478} & {\small 21.39} \\ \hline
 {\small \textbf{Twitter}}& {\small 41,652,230} & {\small 1,202,513,046} & {\small 57.74} \\ \hline
\end{tabular}
\end{table}

\begin{figure*}[!t]
\centering
\subfigure[Facebook]{\includegraphics[width=0.32 \textwidth]{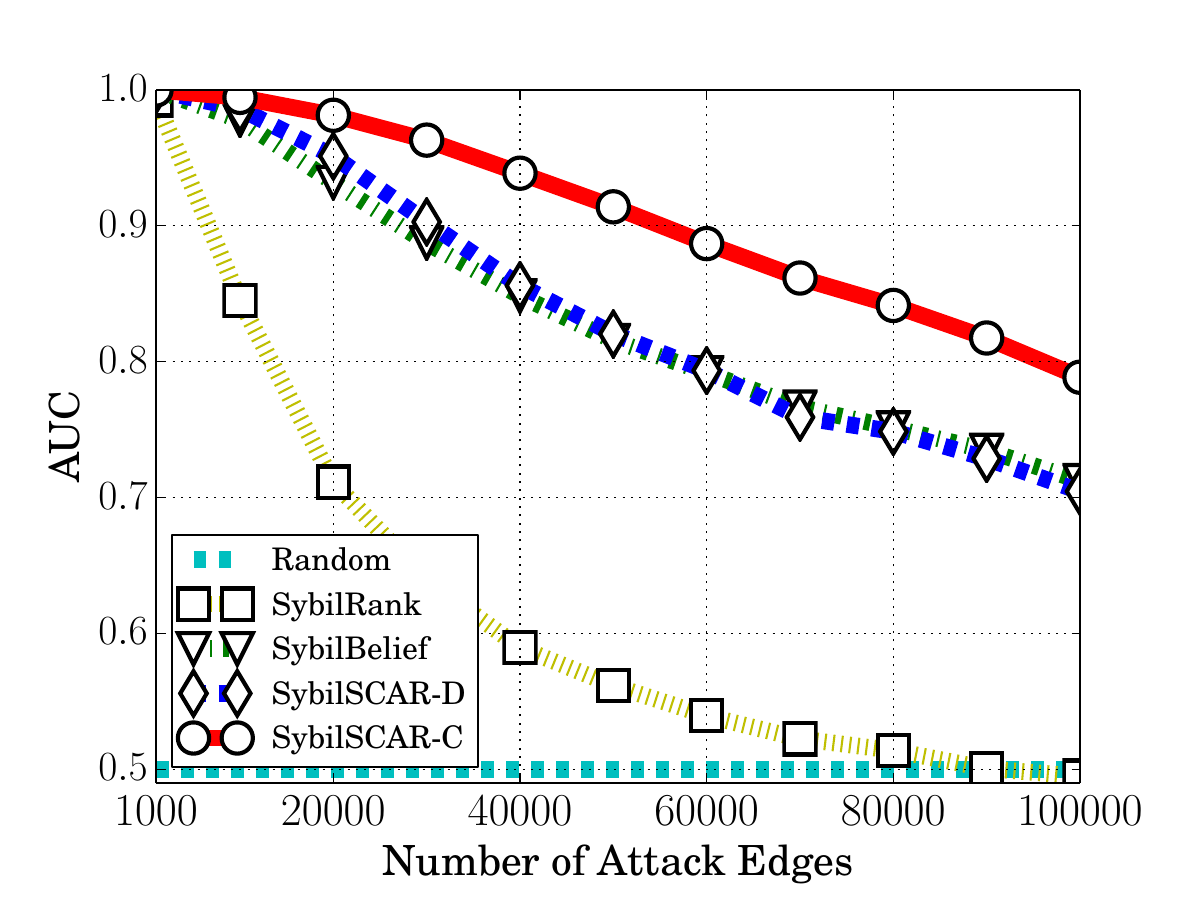} \label{ego-Facebook}}
\subfigure[Enron]{\includegraphics[width=0.32 \textwidth]{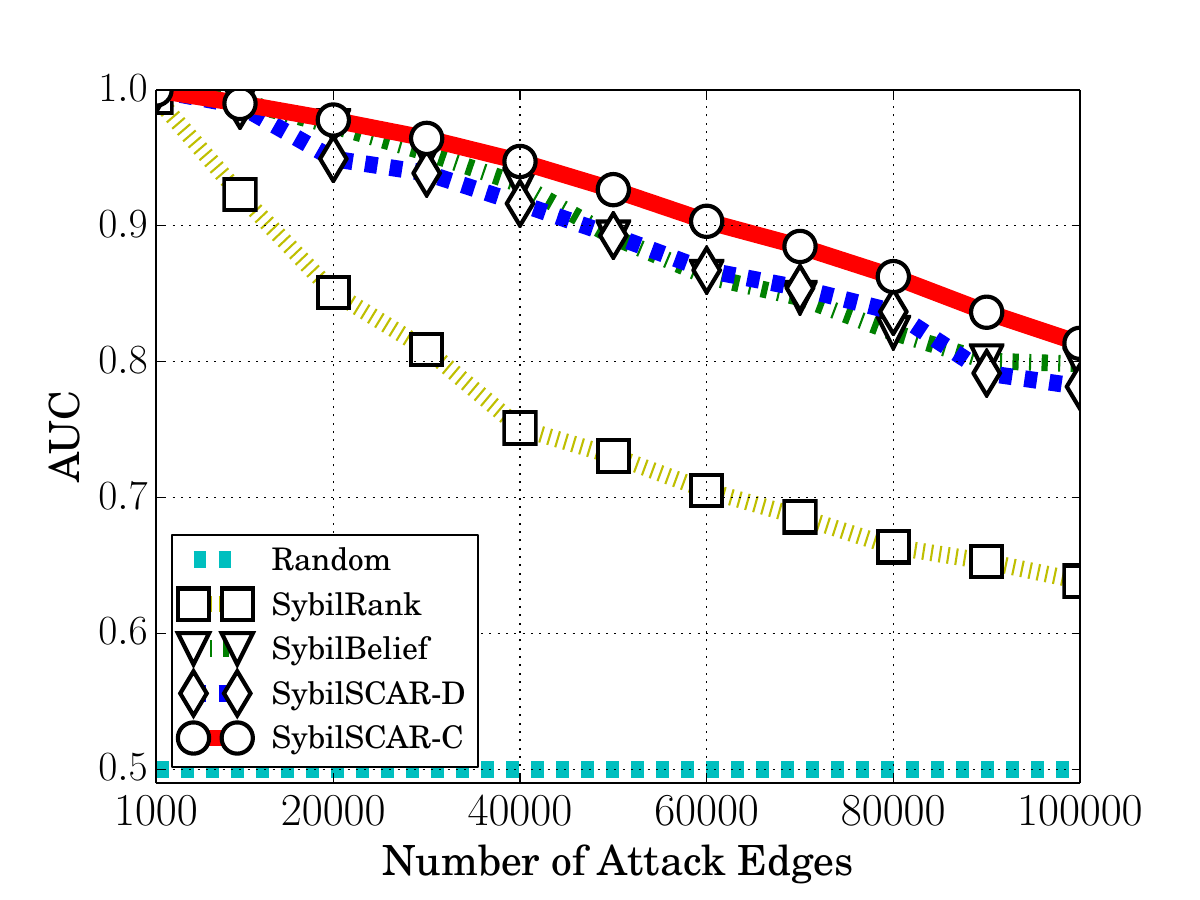} \label{email-Enron}}
\subfigure[Epinions]{\includegraphics[width=0.32 \textwidth]{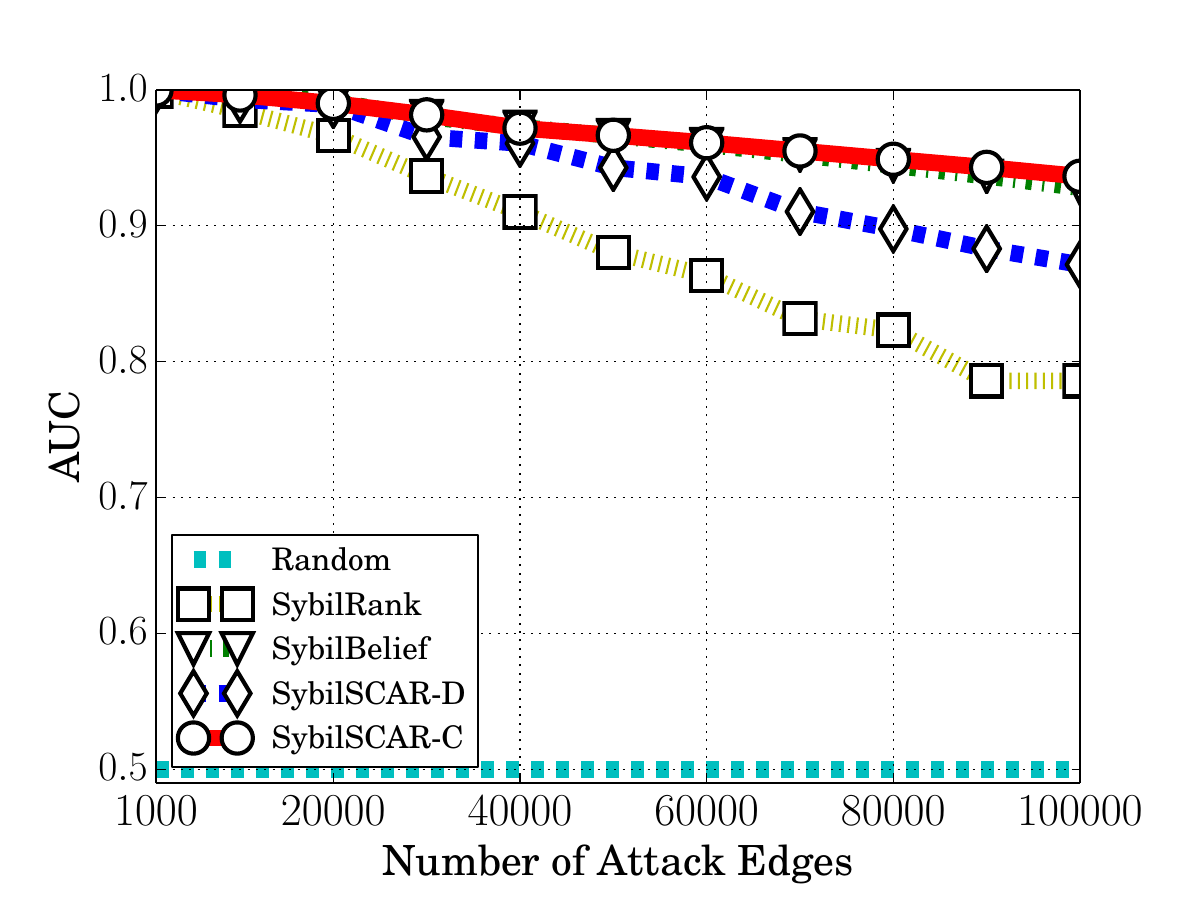} \label{web-Stanford}}
\caption{AUCs of compared methods as the number of attack edges becomes large. SybilSCAR-C and SybilSCAR-D are substantially more accurate than SybilRank, and SybilSCAR-C is slightly more accurate than SybilBelief and SybilSCAR-D, when the number of attack edges is large.} 
\label{auc-synthetic}
\end{figure*}

\subsubsection{Dataset Description}
We use 1) three real-world social networks with synthesized Sybils and 2) a large-scale Twitter dataset with real Sybils for evaluations. 
Table~\ref{dataset_stat} shows some basic statistics about our datasets. 

\myparatight{Social networks with synthesized Sybils} 
We use a real social graph as the benign region while synthesizing the Sybil region and adding attack edges between the two regions uniformly at random. 
There are different ways to synthesize the Sybil region. For instance, we can use a network model (e.g., Preferential Attachment (PA) model~\cite{Barabasi99}) to generate a Sybil region. 
A Sybil region that is synthesized by a network model might be structurally very different from the benign region, e.g.,  although the PA model can generate graphs that have similar degree distribution with real social networks, the generated graphs have very small clustering coefficients, which is very different from real-world social networks. 
Such structural difference could bias Sybil detection results~\cite{alvisiSybil13}.
Moreover, a Sybil region synthesized by a network model like PA does not have community structures, making it unrealistic. 
Therefore, following recent studies~\cite{alvisiSybil13,sybilbelief}, we consider a Sybil attack in which the Sybil region is a replicate of the benign region. 
This way of synthesizing the Sybil region can avoid the structural difference between the two regions, and both Sybil region and benign region have complex community structures. 

We utilize three social networks, i.e., Facebook (4,039 nodes and 88,234 edges), Enron (33,696 nodes and 180,811 edges), and Epinions (75,877 nodes and 811,478 edges), to represent different application scenarios. 
We obtained these datasets from SNAP (http://snap.stanford.edu/ data/index.html).  
A node in Facebook dataset represents a user in Facebook, and two nodes are connected if they are friends.
A node in Enron dataset represents an email address, and an edge between two nodes indicate at least one email was exchanged between the two corresponding email addresses. 
Epinions is a who-trust-whom online social network of a general consumer review site Epinions.com. The nodes in Epinions denote 
members of the site. And in order to maintain quality, Epinsons encourages users to specify which other users they trust, and uses the resulting web of the trust to order the product reviews seem by each person. 
For each social network, we use it as the benign region and replicate it as a Sybil region. Moreover, without otherwise mentioned, we add 1,000 attack edges uniformly at random.

\myparatight{Twitter dataset with real Sybils} 
We obtained a snapshot of a large-scale Twitter follower-followee network crawled by Kwak et al.~\cite{kwak2010twitter}. 
We transformed the follower-followee network into an undirected one via keeping an edge between two users if there are at least one directed edge between them. The undirected Twitter graph has 41,652,230 nodes and 1,202,513,046 edges, with an average degree of 57.74.
To perform evaluation, we need ground truth labels of the users. Since the Twitter network includes users' Twitter IDs, we wrote a crawler to visit each user's profile using Twitter's API, which told us the status (i.e., active, suspended, or deleted)  of each user.
We found that 205,355 users were suspended by Twitter and we treated them as Sybils; 36,156,909 users were still active and we treated them as benign users. The remaining 5,289,966 users were deleted. 
As deleted users could be deleted by Twitter or by users themselves, we could not distinguish the two cases without accessing to Twitter's internal data. Therefore, we treat them as unlabeled users. 
The average number of attack edges per Sybil is 181.55. Therefore, the Twitter network has a very weak homophily.
Note that the number of  benign users and  the number of Sybils are very unbalanced, i.e., the number of labeled benign users is 176 times larger than the number of labeled Sybils.

\myparatight{Training and testing sets} For a social network with synthesized Sybils, we select 200 nodes uniformly at random and use them as a training dataset. For the Twitter dataset, we select 500,000 nodes uniformly at random and use them as a training dataset. The remaining benign and Sybil nodes are used as testing data.

\subsubsection{Compared Methods} 
{We compare SybilSCAR-C and SybilSCAR-D with SybilRank~\cite{sybilrank}, a state-of-the-art RW-based method,  and SybilBelief~\cite{sybilbelief}, a state-of-the-art LBP-based method. In addition, we use random guessing as a baseline.}

For SybilSCAR-C and SybilSCAR-D, we set $\theta=0.1$ to consider possible label noises, i.e., we assign a prior probability 0.6, 0.4, and 0.5 to labeled Sybils, labeled benign users, and unlabeled nodes, respectively. 
We set $\delta=10^{-3}$ and $T=20$. 
Considering different average degrees of Facebook, Enron, Epinions, and Twitter, we set $\hat{w}=0.01$, $0.04$, $0.02$, and $0.01$ for SybilSCAR-C, respectively.  
We set the parameters of SybilRank and SybilBelief according to the papers that introduced them. For instance, for SybilBelief, the edge weight is set to be 0.9 for all edges; 
 SybilRank requires early termination, and we set the number of iterations as $\lceil \log(|V|) \rceil$.

We implemented SybilRank, SybilSCAR-C, and SybilSCAR-D in C++. We obtained a basic implementation of SybilBelief (also in C++) from its authors and optimized the implementation. We performed all our experiments on a Linux machine with 16GB memory and 8 cores.

\subsection{Ranking Accuracy}

\begin{figure}[!t]
	\centering
	\includegraphics[width=0.45\textwidth]{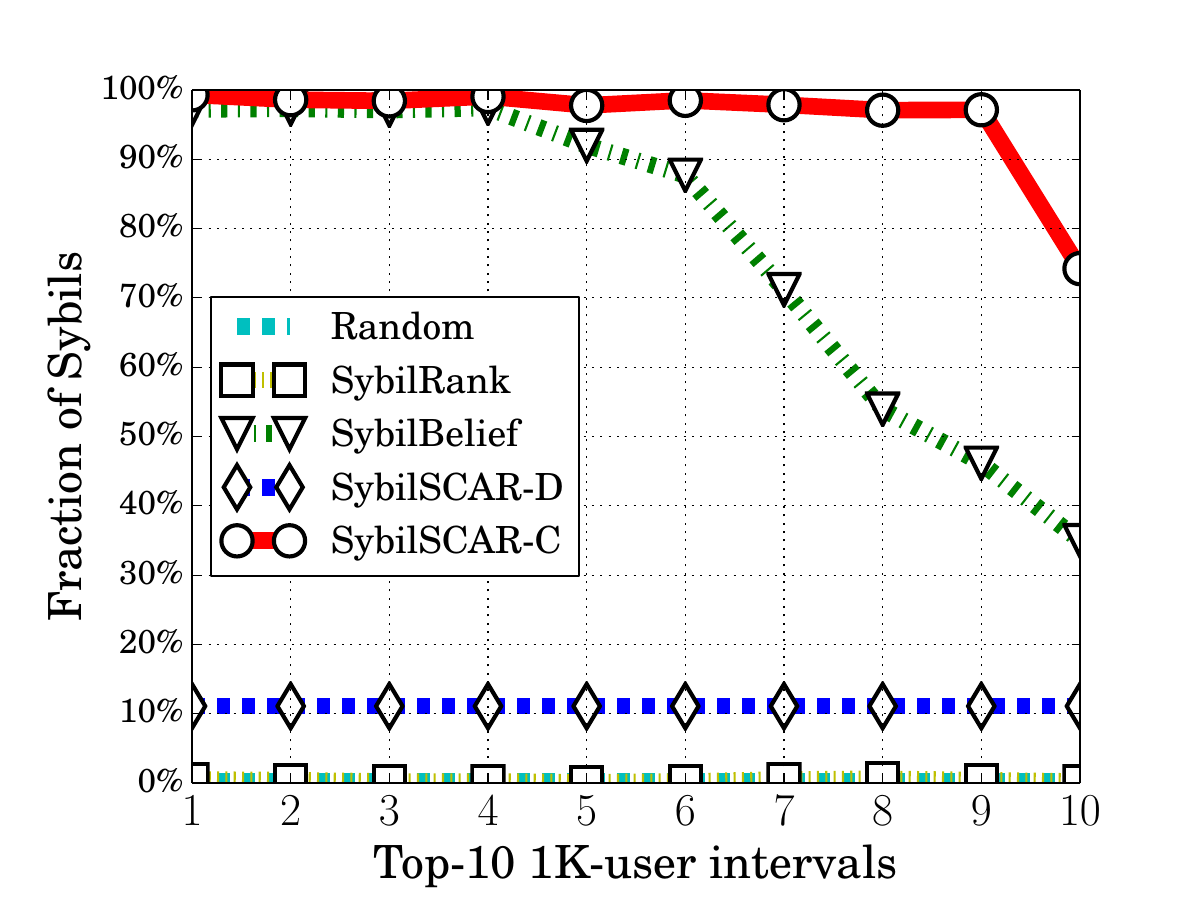}
	\caption{Fraction of Sybils in top ranked intervals on the Twitter dataset. SybilSCAR-C performs better than SybilBelief and SybilSCAR-D, while SybilSCAR-D and SybilBelief perform much better than SybilRank.}
	\label{ranking_large}
\end{figure}

\begin{figure*}[!t]
\centering
\subfigure[Facebook]{\includegraphics[width=0.45\textwidth]{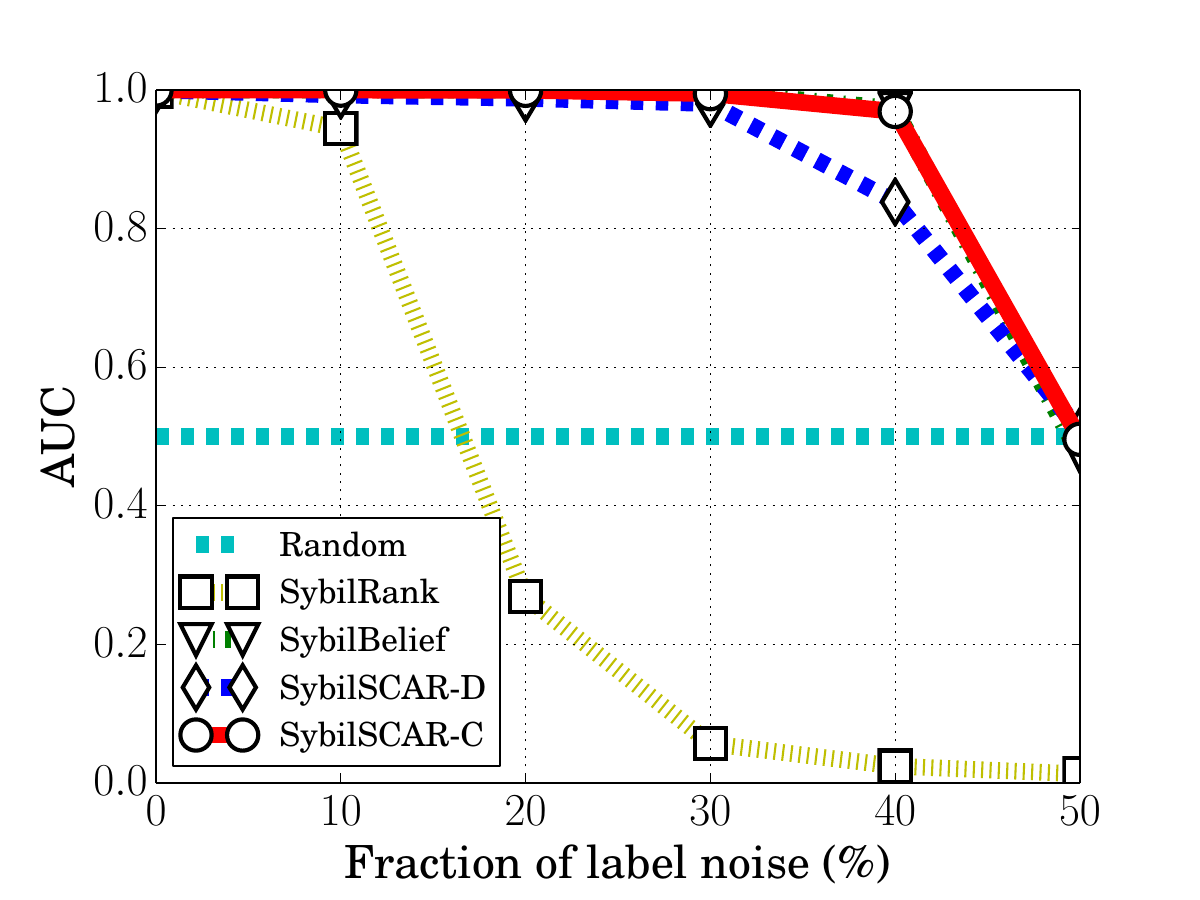}\label{facebook}}
\subfigure[Large Twitter]{\includegraphics[width=0.45\textwidth]{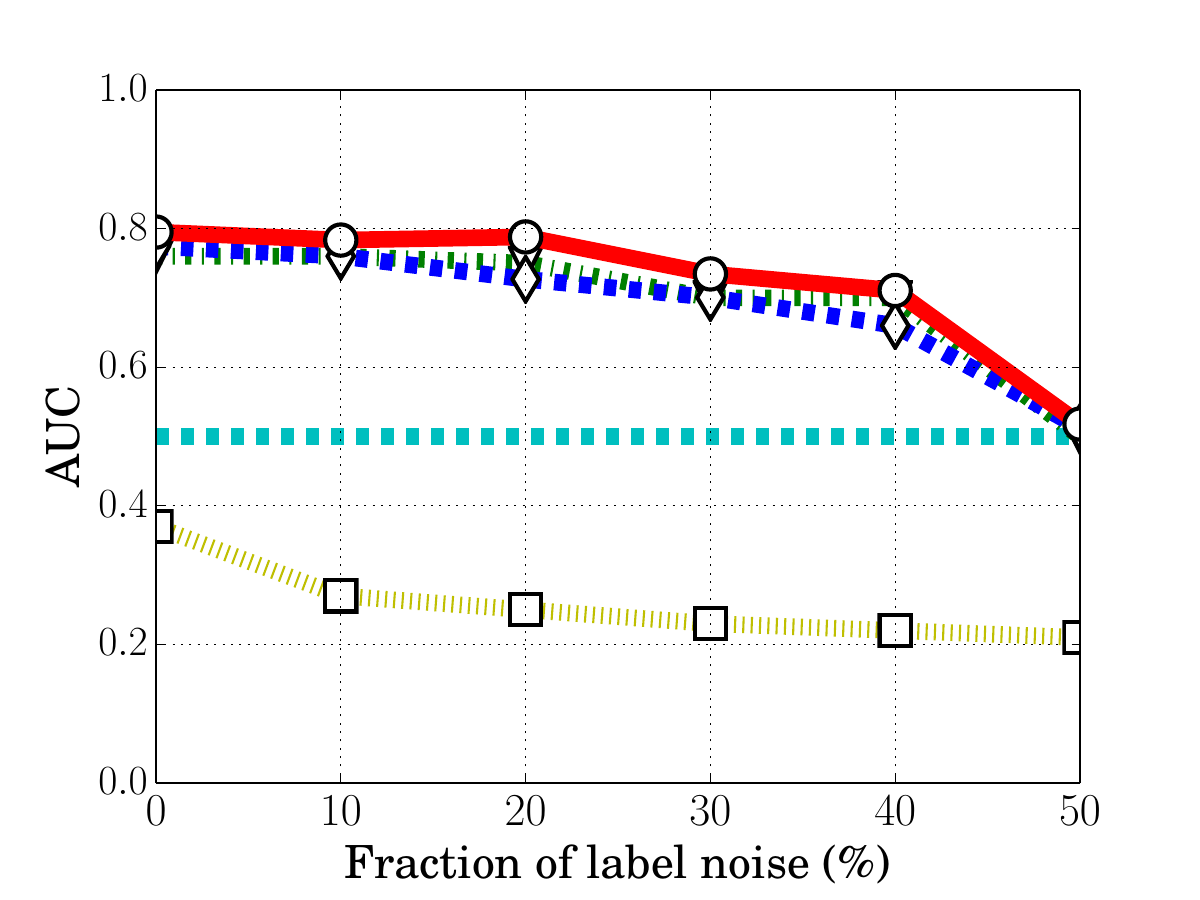}\label{largetwitter}}
\caption{AUCs of SybilRank, SybilBelief, SybilSCAR-D, and SybilSCAR-C vs. level of label noise. SybilSCAR-C is slightly better than SybilSCAR-D and SybilBelief, while they are much more robust to label noise than SybilRank. Note that SybilBelief and SybilSCAR-C overlap in (a).} 
\label{labelnoises-auc}
\end{figure*}

\subsubsection{Results on Social Networks with Synthesized Sybils}
Viswanath et al.~\cite{Viswanath10} demonstrated that Sybil detection methods can be treated as ranking mechanisms, and they can be evaluated using Area Under the Receiver Operating Characteristic Curve (AUC).
Therefore, we adopt AUC  to evaluate ranking accuracy. 
Suppose we rank nodes with respect to their posterior reputation/probability of being a Sybil in a descending order. AUC is the probability that a randomly selected Sybil ranks higher than a randomly selected benign node. Note that random guessing, which ranks all nodes uniformly at random, has an AUC of 0.5. 

Figure~\ref{auc-synthetic} shows AUCs of the compared methods as we increase the number of attack edges from 1,000 to 100,000. 
We have three observations.  
First, when a social network has strong homophily, i.e., the number of attack edges is small,  all the compared methods achieve very high AUCs. For instance, SybilRank, SybilBelief,  SybilSCAR-C, and SybilSCAR-D all achieve AUCs that are close to 1 when the number of attack edges is less than 1,000.
Second, SybilSCAR-C, SybilSCAR-D, and SybilBelief are substantially more accurate than SybilRank when the number of attack edges becomes large, i.e., the social networks have weak homophily. 
{A possible reason is that SybilSCAR-C, SybilSCAR-D, and SybilBelief can leverage both labeled benign users and labeled Sybils in the training dataset.}
{Third, SybilSCAR-C achieves slightly larger AUCs than SybilBelief and SybilSCAR-D.  
 Compared with SybilBelief, SybilSCAR-C uses a new neighbor influence by directly modeling the homophily property of the social network. 
 Compared with SybilSCAR-D, SybilSCAR-C uses a constant weight for all edges, which may make the labeled nodes have larger influence to their neighbors.  
}

\subsubsection{Results on the Large-scale Twitter Dataset}

{The AUCs of SybilRank, SybilBelief, SybilSCAR-D, and SybilSCAR-C are 0.37, 0.76, 0.77, and 0.80, respectively. SybilRank performs worse than random guessing. SybilSCAR-C performs better than SybilSCAR-D, which is comparable with SybilBelief. }  
These results are consistent with those in social networks with synthesized Sybils, because the number of attack edges in the Twitter dataset is very large. 
We note that these AUCs are obtained via using a \emph{balanced} training dataset. Specifically, among the 500,000 nodes in the training dataset, benign nodes are much more than Sybils; we subsample some benign nodes such that we have the same number of  benign nodes and  Sybils, and we use them as a balanced training dataset.  We found that all the methods have very low AUCs (worse than random guessing) if we use the original unbalanced training dataset consisting of the randomly sampled 500,000 nodes. It would be an interesting future work to theoretically understand the impact of balanced/unbalanced training dataset on the accuracy of these methods. 

In practice, the ranking of users can be used as a priority list to guide human workers to manually inspect users and detect Sybils. In particular, inspecting users according to their rankings could aid human workers to detect more Sybils than inspecting randomly picked users, within the same amount of time. When ranking is used for such purpose, the number of Sybils in top-ranked users is important because human workers can only inspect a limited number of users.
AUC measures the overall ranking performance, but it cannot tell Sybils among the top-ranked users. Therefore, we further compare the considered methods using the fraction of Sybils in top-ranked users. 

{Specifically, for each method, we divide the top-10K users obtained by the method into 10 intervals, where each interval has 1K users.  Figure~\ref{ranking_large} shows the fraction of Sybils in each 1K-user interval for the compared methods. 
First, SybilSCAR-C performs better than SybilBelief. Specifically, the fraction of Sybils ranges from 74.2\% to 99.3\% in the top-10 1K-user intervals for SybilSCAR-C, while the range is from 35.0\% to 97.4\% for SybilBelief. 
Second, SybilSCAR-C and SybilBelief outperform SybilSCAR-D, which demonstrates that the predefined constant edge weight is more informative than degree-normalized edge weight for ranking Sybils. 
Third, SybilRank is close to random guessing. } 

We note that SybilWalk~\cite{SybilWalk} was shown to achieve good results on the same Twitter dataset. However, SybilWalk heavily preprocessed the Twitter dataset (e.g., identifying the high-degree nodes that have many attack edges and removing them) to significantly reduce the number of attack edges. We did not perform such preprocessing since it is time-consuming to identify such nodes in practice. Without such preprocessing, SybilWalk has similar performance with SybilSCAR-D.

\begin{figure*}[!t]
\centering
\subfigure[Peak memory usage]{\includegraphics[width=0.45\textwidth]{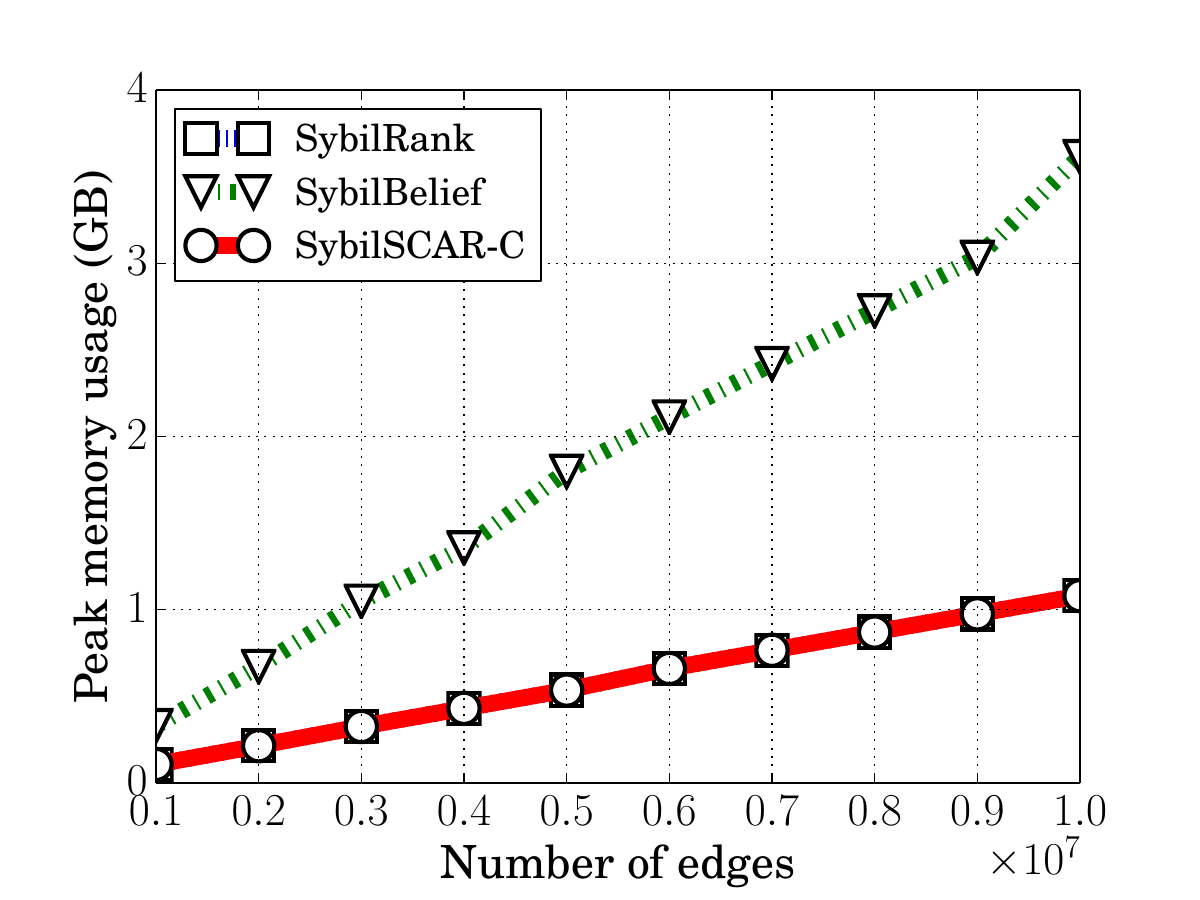}\label{efficiency-space}}
\subfigure[Time]{\includegraphics[width=0.45\textwidth]{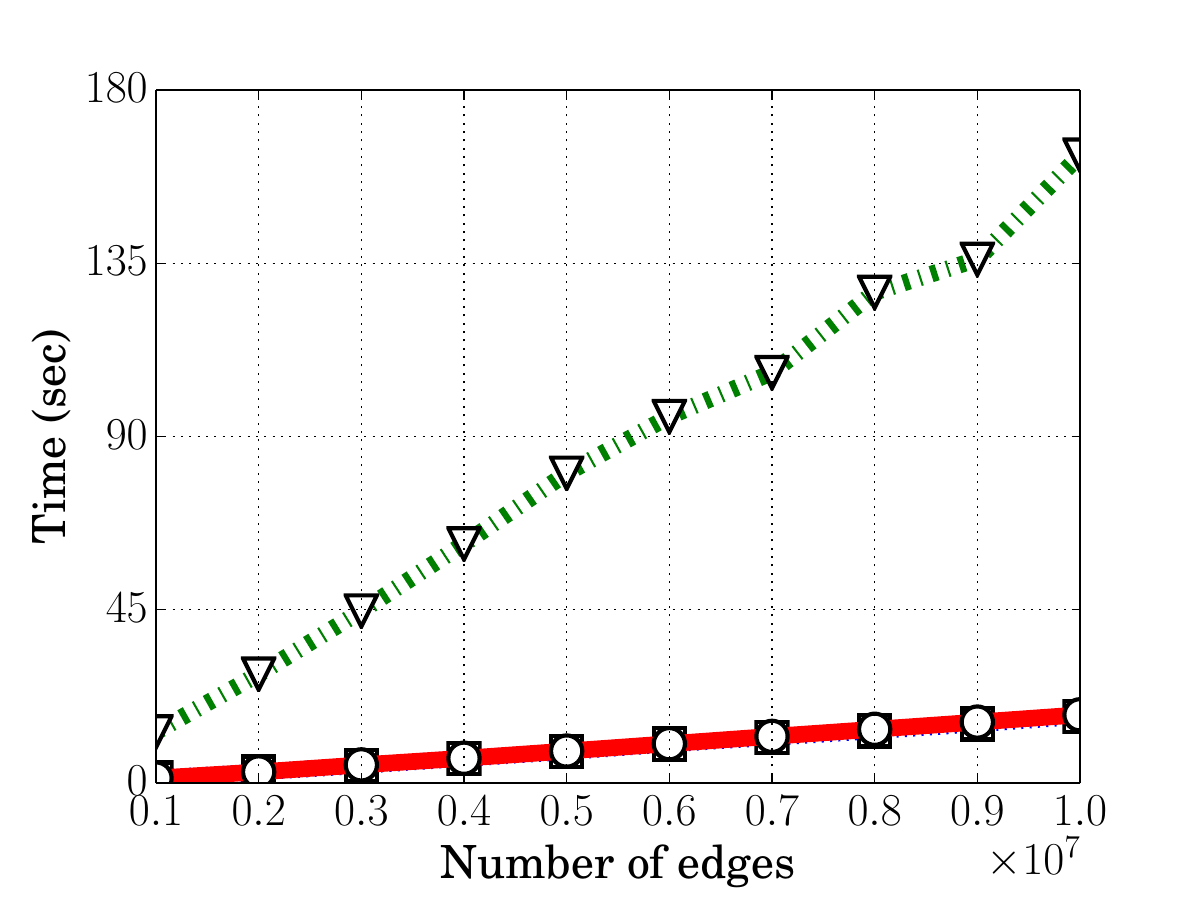}\label{efficiency-time}}
\caption{Space and time efficiency of SybilRank, SybilBelief, and SybilSCAR-C (SybilSCAR-D has the same space and time efficiency with SybilSCAR-C), vs. number of edges. 
SybilSCAR-C and SybilRank have almost the same space and time efficiency, while SybilSCAR-C is several times more space efficient and one order of magnitude more time efficient than SybilBelief.}
\label{efficiency}
\end{figure*}

\subsection{Robustness to label noise} 
In practice, a training dataset might have noises, i.e., some labeled benign users are actually Sybils and some labeled Sybils are actually benign. 
Such noises could be introduced by human mistakes~\cite{wang2012social}.
Thus, one natural question is how label noise impacts the accuracy of detection methods.

For a given level of noise $\tau\%$, we randomly choose $\tau\%$ of labeled Sybils in the training dataset and mislabel them as benign users; and we also sample $\tau\%$ of labeled benign users in the training dataset and mislabel them as Sybils. We vary $\tau\%$ from $10\%$ to $50\%$ with a step size of $10\%$. Note that we didn't perform experiments for $\tau\% > 50\%$ as all these methods cannot detect Sybils when a majority of labels are incorrect.  
{Figure~\ref{labelnoises-auc} shows the AUCs of SybilRank, SybilBelief, SybilSCAR-C, and SybilSCAR-D on Facebook (note that we have similar results on Enron and Epinions, and thus omit them for simplicity) and Twitter datasets
against different levels of label noises.  
We observe that 1) SybilSCAR-C has the best robustness against label noise; 
2) SybilBelief is slightly more robust than SybilSCAR-D to label noise;  
3) SybilSCAR-C, SybilSCAR-D, and SybilBelief are more robust to label noise than SybilRank. For instance, on the Facebook dataset, SybilSCAR-C and SybilBelief can tolerate label noise up to 40\%, SybilSCAR-D can tolerate label noise up to 30\%, while SybilRank performs worse than random guessing when label noise is higher than 20\%. 
} 

We believe it is an interesting future work to theoretically understand the robustness to label noise of different methods. In the following, we provide a possible explanation on why  SybilSCAR-C, SybilSCAR-D, and SybilBelief are more robust to label noise than SybilRank. 
When there are label noises, some benign nodes are treated as Sybils. The edges between these nodes and the rest of benign nodes become attack edges, while the original attack edges that connect with these nodes become edges in the new Sybil region. Similarly, some Sybils are mislabeled as benign nodes. The edges between these nodes and the rest of Sybils are treated as attack edges, while the original attack edges that connect with these nodes become edges in the new benign region. Since we randomly sample mislabeled nodes, the new attack edges  are likely to be more than the original attack edges that become edges within the new benign region or Sybil region. Therefore, SybilSCAR-C, SybilSCAR-D, and SybilBelief outperform SybilRank, because they can tolerate a larger number of attack edges. Moreover, when label noise is larger than a certain threshold (this threshold is graph-dependent), the new attack edges are more than what SybilSCAR-C, SybilSCAR-D, and SybilBelief can tolerate, and thus their performance degrades significantly.

\subsection{Scalability}
We evaluate scalability in terms of the peak memory and time used by each method. Because evaluating scalability requires social networks with varying number of edges, we evaluate scalability on synthesized graphs with different number of edges. 
Note that our purpose here is not to concern about the accuracy, which depends on the number of iterations of each method. Thus, to avoid the bias introduced by the number of iterations, we run all methods with the same number of iterations. 

Figure~\ref{efficiency} exhibits the peak memory and time used by SybilRank, SybilBelief, and SybilSCAR-C  (SybilSCAR-D has the same complexity with SybilSCAR-C, and thus we omit its results for simplicity) for different number of edges with 20 iterations. 
We observe that: 1) all methods have linear space and time complexity, which is consistent with our theoretical analysis in Section~\ref{sec:complexity};
2) SybilRank and SybilSCAR-C use almost the same space and time; 
3) SybilSCAR-C requires a few times less memory than SybilBelief and is one order of magnitude faster than SybilBelief.  
The reason is that SybilBelief needs a large amount of resources to store and maintain the neighbor influence on every edge. 
We note that we optimized the implementation of SybilBelief provided by its authors, and our optimized version is one order of magnitude faster than the unoptimized version. 
Using the unoptimized implementation, SybilSCAR is around two orders of magnitude faster than SybilBelief, which was reported in~\cite{sybilscar}.

\subsection{Convergence}
We define a relative error of residual posterior probability vectors of SybilSCAR-C (or SybilSCAR-D) as $\frac{\|\hat{\mathbf{p}}^{(t)} - \hat{\mathbf{p}}^{(t-1)}\|_1} {\|\hat{\mathbf{p}}^{(t)}\|_1}$, where $\hat{\mathbf{p}}^{(t)}$ is the residual vector of posterior probability produced by SybilSCAR-C (or SybilSCAR-D) in the $t$th iteration. Similarly, we can define relative errors for SybilRank and SybilBelief using their vectors of posterior reputation/probability. 
{Figure~\ref{convergence} shows the relative errors vs. the number of iterations on Facebook with 1,000 attack edges. We observe that 1) SybilSCAR-C, SybilSCAR-D, and SybilRank converge after several iterations; 
and 2) the relative errors of SybilBelief oscillate. 
SybilBelief does not converge because there exists many loops in real-world social networks and  LBP may oscillate on  graphs with loops,  as pointed out by the author of LBP~\cite{Pearl88}.  
}

\begin{figure}[t]
	\centering
	\includegraphics[width=0.45\textwidth]{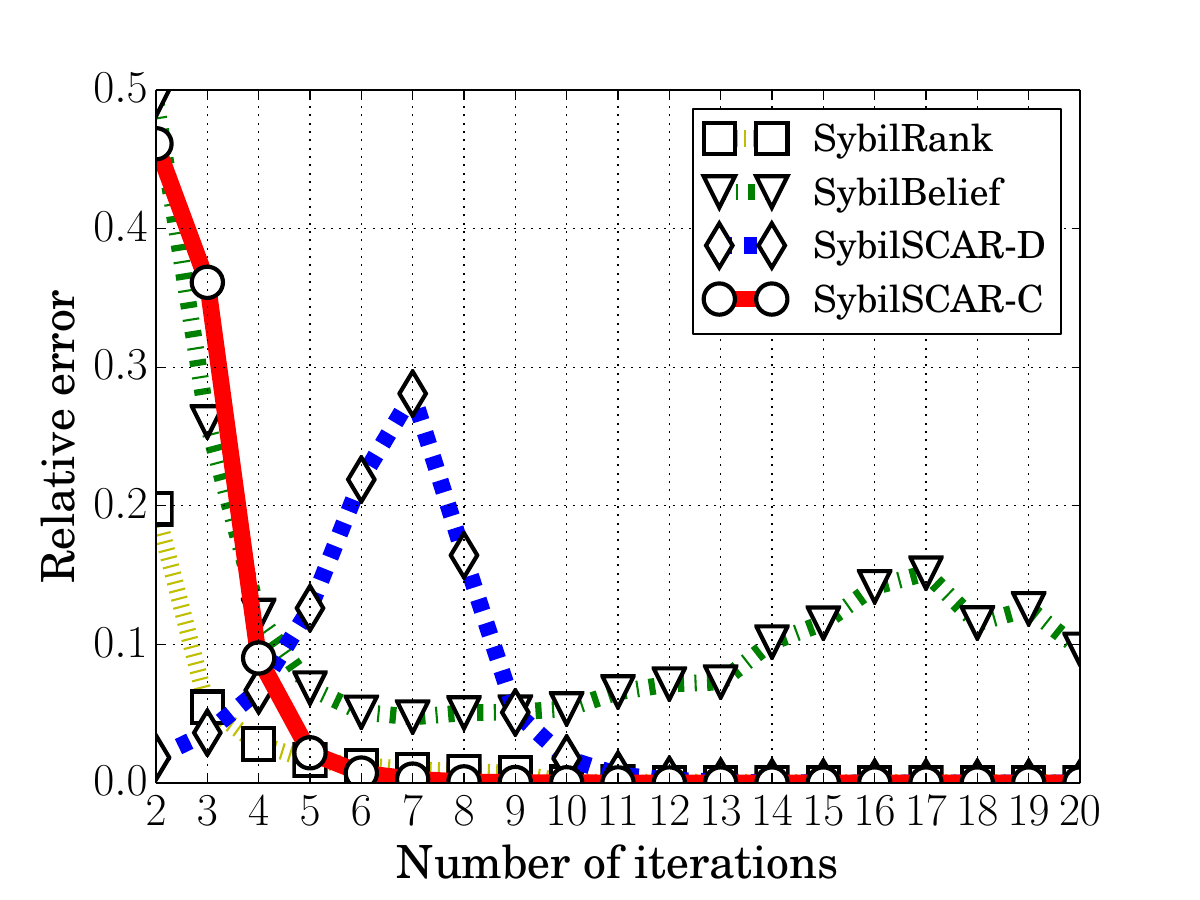}
	\caption{Relative errors of SybilRank, SybilBelief, SybilSCAR-D, and SybilSCAR-C vs. the number of iterations on Facebook with 1,000 attack edges. SybilRank, SybilSCAR-D, and SybilSCAR-C can converge, but SybilBelief cannot.}
	\label{convergence}
\end{figure}

\begin{figure*}[!t]
\centering
\subfigure[SybilSCAR-D]{\includegraphics[width=0.45\textwidth]{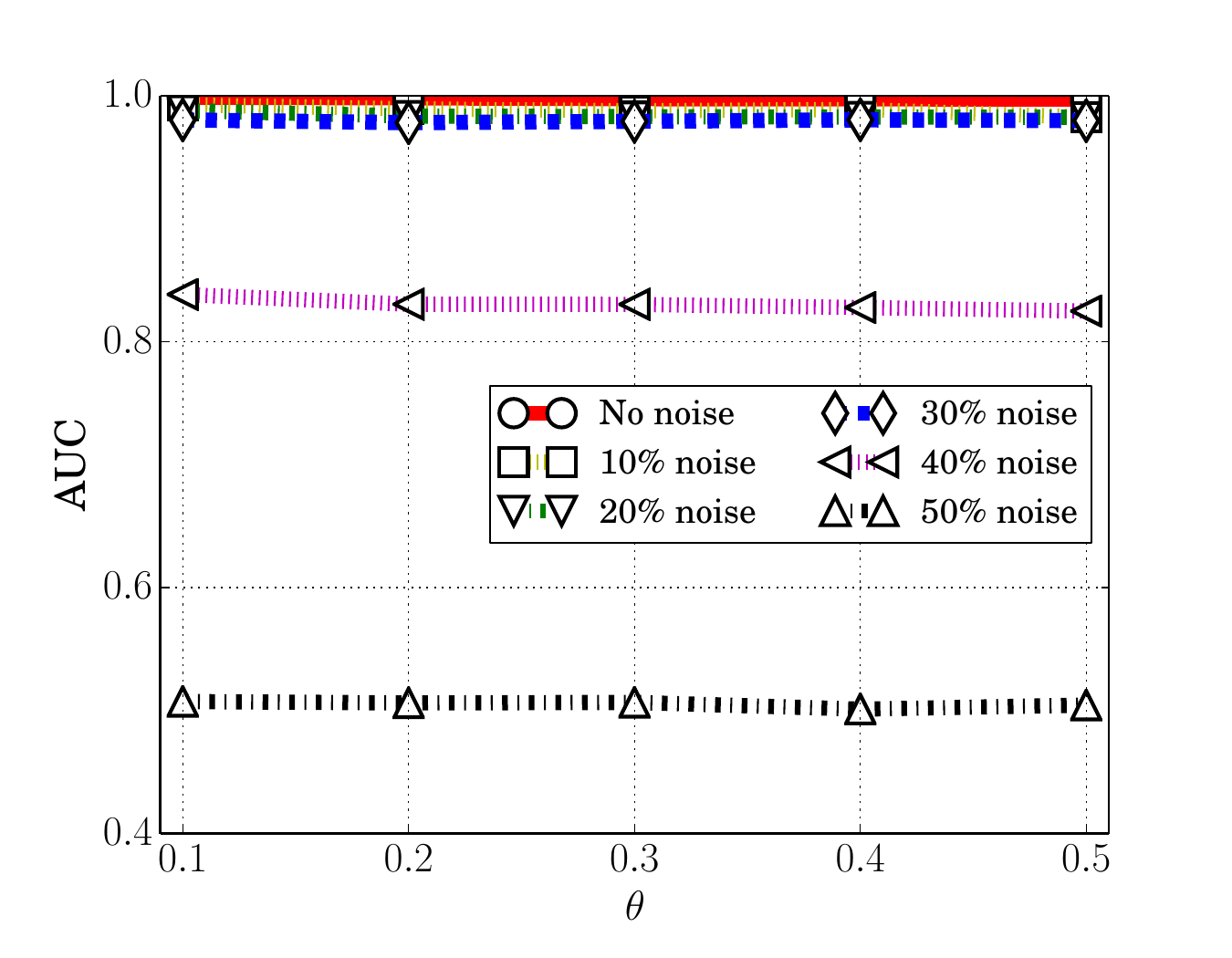}}
\subfigure[SybilSCAR-C]{\includegraphics[width=0.45\textwidth]{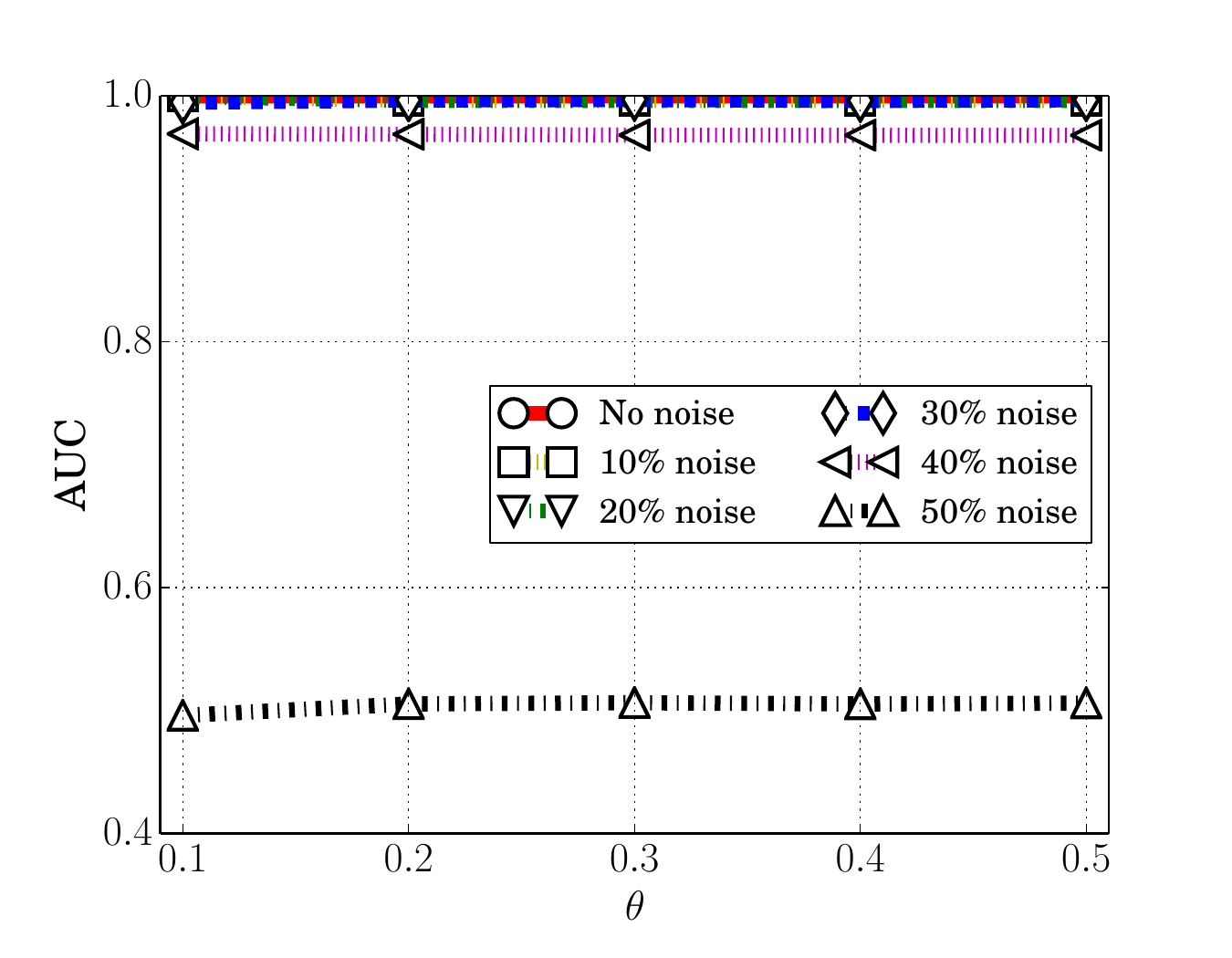}}
\caption{AUCs of (a) SybilSCAR-D and (b) SybilSCAR-C for different levels of label noise and different values for the parameter $\theta$.}
\label{theta}
\end{figure*}

\subsection{Impact of the Parameter $\theta$}
The parameter $\theta$ is the residual prior probability of labeled nodes. Figure~\ref{theta} shows the AUCs of SybilSCAR-D and SybilSCAR-C for different $\theta$ and different levels of label noise, where the dataset is Facebook. Note that $\theta$ should be in the range (0, 0.5]. Therefore, we explored the values 0.1, 0.2, 0.3, 0.4, and 0.5. We observe that both SybilSCAR-D and SybilSCAR-C are stable with respect to the choice of $\theta$.

\subsection{Summary} 
We summarize our key observations as follows: 
{
\begin{itemize}
\item Compared to SybilRank, SybilSCAR-C and SybilSCAR-D are substantially more accurate and more robust to label noise. 
\item Compared to SybilBelief, SybilSCAR-C is more accurate, significantly more scalable, and guaranteed to converge.
\item SybilSCAR-C outperforms SybilSCAR-D, showing that a constant edge weight is more informative than degree-normalized edge weight for Sybil detection. 
\item SybilSCAR-C and SybilSCAR-D are stable with respect to the prior probabilities of the labeled nodes. 
\end{itemize}
}

\section{Conclusion and Future Work}

In this work, we first propose a local rule based framework to unify state-of-the-art Random Walk (RW)-based  
and Loopy Belief Propagation (LBP)-based Sybil detection methods. 
Our framework makes it possible to analyze and compare different Sybil detection methods in a unified way. 
Second, we design a new local rule. Our local rule integrates advantages of RW-based methods and LBP-based methods, while overcoming their limitations. 
Third, we perform both theoretical and empirical evaluations.
Theoretically, SybilSCAR has a tighter asymptotical bound on the number of Sybils that are falsely accepted into the social network than existing structure-based methods. Moreover, SybilSCAR can guarantee to converge. 
Empirically, our experimental results on both synthesized Sybils and real-world Sybils demonstrate that SybilSCAR is more accurate and more robust to label noise than SybilRank, while SybilSCAR is more accurate and significantly more scalable than SybilBelief.  

Future research directions include 1) learning the homophily strength for each edge; 
2) theoretically analyzing different local rules with respect to accuracy and robustness to label noise; 
3) theoretically understanding the impacts of balanced/unbalanced training dataset; 
and 4) applying SybilSCAR to detect other types of Sybils such as web spams, fake reviews, fake likes, etc.

\bibliographystyle{IEEEtran}
\bibliography{refs}

\begin{thebibliography}{10}
\providecommand{\url}[1]{#1}
\csname url@samestyle\endcsname
\providecommand{\newblock}{\relax}
\providecommand{\bibinfo}[2]{#2}
\providecommand{\BIBentrySTDinterwordspacing}{\spaceskip=0pt\relax}
\providecommand{\BIBentryALTinterwordstretchfactor}{4}
\providecommand{\BIBentryALTinterwordspacing}{\spaceskip=\fontdimen2\font plus
\BIBentryALTinterwordstretchfactor\fontdimen3\font minus
  \fontdimen4\font\relax}
\providecommand{\BIBforeignlanguage}[2]{{%
\expandafter\ifx\csname l@#1\endcsname\relax
\typeout{** WARNING: IEEEtran.bst: No hyphenation pattern has been}%
\typeout{** loaded for the language `#1'. Using the pattern for}%
\typeout{** the default language instead.}%
\else
\language=\csname l@#1\endcsname
\fi
#2}}
\providecommand{\BIBdecl}{\relax}
\BIBdecl

\bibitem{gao2012towards}
H.~Gao, Y.~Chen, K.~Lee, D.~Palsetia, and A.~Choudhary, ``Towards online spam
  filtering in social networks,'' in \emph{NDSS}, 2012.

\bibitem{alexaStat}
\BIBentryALTinterwordspacing
{Facebook Popularity.} (2015, October). [Online]. Available:
  \url{http://www.alexa.com/topsites}
\BIBentrySTDinterwordspacing

\bibitem{Twittersybil}
{Sybils in Twitter},
  ``\url{http://www.nbcnews.com/technology/1-10-twitter-accounts-fake-say-researchers-2d11655362}.''

\bibitem{election}
\BIBentryALTinterwordspacing
{Hacking Election.} (2016, May). [Online]. Available:
  \url{http://goo.gl/G8o9x0}
\BIBentrySTDinterwordspacing

\bibitem{stock}
\BIBentryALTinterwordspacing
{Hacking Financial Market.} (2016, May). [Online]. Available:
  \url{http://goo.gl/4AkWyt}
\BIBentrySTDinterwordspacing

\bibitem{Thomas11}
K.~Thomas, C.~Grier, J.~Ma, V.~Paxson, and D.~Song, ``Design and evaluation of
  a real-time url spam filtering service,'' in \emph{IEEE S \& P}, 2011.

\bibitem{Bilge09}
L.~Bilge, T.~Strufe, D.~Balzarotti, and E.~Kirda, ``All your contacts are
  belong to us: Automated identity theft attacks on social networks,'' in
  \emph{WWW}, 2009.

\bibitem{sybilrank}
Q.~Cao, M.~Sirivianos, X.~Yang, and T.~Pregueiro, ``Aiding the detection of
  fake accounts in large scale social online services,'' in \emph{NSDI}, 2012.

\bibitem{Yu06}
H.~Yu, M.~Kaminsky, P.~B. Gibbons, and A.~Flaxman, ``Sybilguard: defending
  against sybil attacks via social networks,'' in \emph{ACM SIGCOMM}.\hskip 1em
  plus 0.5em minus 0.4em\relax ACM, 2006.

\bibitem{Yu08}
H.~Yu, P.~B. Gibbons, M.~Kaminsky, and F.~Xiao, ``{SybilLimit}: A near-optimal
  social network defense against {Sybil} attacks,'' in \emph{IEEE S \& P},
  2008.

\bibitem{Danezis09}
G.~Danezis and P.~Mittal, ``{SybilInfer}: Detecting {Sybil} nodes using social
  networks,'' in \emph{NDSS}, 2009.

\bibitem{Mohaisen11}
A.~Mohaisen, N.~Hopper, and Y.~Kim, ``Keep your friends close: Incorporating
  trust into social network-based sybil defenses,'' in \emph{IEEE INFOCOM},
  2011.

\bibitem{Yang12-spam}
C.~Yang, R.~Harkreader, J.~Zhang, S.~Shin, and G.~Gu, ``Analyzing spammer's
  social networks for fun and profit,'' in \emph{WWW}, 2012.

\bibitem{integro}
Y.~Boshmaf, D.~Logothetis, G.~Siganos, J.~Leria, J.~Lorenzo, M.~Ripeanu, and
  K.~Beznosov, ``Integro: Leveraging victim prediction for robust fake account
  detection in osns,'' in \emph{NDSS}, 2015.

\bibitem{smartwalk}
Y.~Liu, S.~Ji, and P.~Mittal, ``Smartwalk: Enhancing social network security
  via adaptive random walks,'' in \emph{ACM CCS}, 2016.

\bibitem{SybilWalk}
J.~Jia, B.~Wang, and N.~Z. Gong, ``Random walk based fake account detection in
  online social networks,'' in \emph{DSN}, 2017.

\bibitem{zhang2016truetop}
J.~Zhang, R.~Zhang, J.~Sun, Y.~Zhang, and C.~Zhang, ``Truetop: A
  sybil-resilient system for user influence measurement on twitter,''
  \emph{IEEE/ACM ToN}, 2016.

\bibitem{sybilbelief}
N.~Z. Gong, M.~Frank, and P.~Mittal, ``Sybilbelief: A semi-supervised learning
  approach for structure-based sybil detection,'' \emph{IEEE TIFS}, vol.~9,
  no.~6, pp. 976--987, 2014.

\bibitem{sybilframe}
P.~Gao, B.~Wang, N.~Z. Gong, S.~Kulkarni, and P.~Mittal, ``Sybilfuse: Combining
  local attributes with global structure to perform robust sybil detection,''
  \emph{MIS2}, 2018.

\bibitem{robustspammer}
H.~Fu, X.~Xie, Y.~Rui, N.~Z. Gong, G.~Sun, and E.~Chen, ``Robust spammer
  detection in microblogs: Leveraging user carefulness,'' \emph{ACM TIST},
  2017.

\bibitem{gang}
B.~Wang, N.~Z. Gong, and H.~Fu, ``Gang: Detecting fraudulent users in online
  social networks via guilt-by-association on directed graphs,'' in
  \emph{ICDM}, 2017.

\bibitem{Pearl88}
J.~Pearl, \emph{Probabilistic reasoning in intelligent systems: networks of
  plausible inference}, 1988.

\bibitem{wang2012social}
G.~Wang, M.~Mohanlal, C.~Wilson, X.~Wang, M.~Metzger, H.~Zheng, and B.~Y. Zhao,
  ``Social turing tests: Crowdsourcing sybil detection,'' \emph{NDSS}, 2013.

\bibitem{Yang11-sybil}
Z.~Yang, C.~Wilson, X.~Wang, T.~Gao, B.~Y. Zhao, and Y.~Dai, ``Uncovering
  social network {Sybils} in the wild,'' in \emph{ACM IMC}, 2011.

\bibitem{Wang10}
A.~H. Wang, ``Don't follow me - spam detection in twitter,'' in \emph{SECRYPT
  2010}, 2010.

\bibitem{yardi10}
G.~S. S.~Yardi, D.~Romero and D.~Boyd, ``Detecting spam in a {Twitter}
  network,'' \emph{First Monday}, vol. 15(1), 2010.

\bibitem{LeeUncovering10}
K.~Lee, J.~Caverlee, and S.~Webb, ``Uncovering social spammers: Social
  honeypots + machine learning,'' in \emph{ACM SIGIR}, 2010.

\bibitem{benevenuto2010detecting}
F.~Benevenuto, G.~Magno, T.~Rodrigues, and V.~Almeida, ``Detecting spammers on
  twitter,'' in \emph{CEAS}, 2010.

\bibitem{Song11}
J.~Song, S.~Lee, and J.~Kim, ``Spam filtering in {Twitter} using
  sender-receiver relationship,'' in \emph{RAID}, 2011.

\bibitem{facebookImmune}
T.~Stein, E.~Chen, and K.~Mangla, ``Facebook immune system,'' in \emph{SNS},
  2011.

\bibitem{Wang13Clickstream}
G.~Wang, T.~Konolige, C.~Wilson, and X.~Wang, ``You are how you click:
  Clickstream analysis for sybil detection,'' in \emph{Usenix Security}, 2013.

\bibitem{CaoCCS14}
Q.~Cao, X.~Yang, J.~Yu, and C.~Palow, ``Uncovering large groups of active
  malicious accounts in online social networks,'' in \emph{ACM CCS}, 2014.

\bibitem{cao2015combating}
Q.~Cao, M.~Sirivianos, X.~Yang, and K.~Munagala, ``Combating friend spam using
  social rejections,'' in \emph{ICDCS}.\hskip 1em plus 0.5em minus 0.4em\relax
  IEEE, 2015, pp. 235--244.

\bibitem{Wang13}
G.~Wang, M.~Mohanlal, C.~Wilson, X.~Wang, M.~Metzger, H.~Zheng, and B.~Y. Zhao,
  ``Social turing tests: Crowdsourcing {Sybil} detection,'' in \emph{NDSS},
  2013.

\bibitem{DanezisAnonyComPET10}
G.~Danezis, C.~Diaz, C.~Troncoso, and B.~Laurie, ``Drac: An architecture for
  anonymous low-volume communications,'' in \emph{PETS}, 2010.

\bibitem{yang2014uncovering}
Z.~Yang, C.~Wilson, X.~Wang, T.~Gao, B.~Y. Zhao, and Y.~Dai, ``Uncovering
  social network sybils in the wild,'' \emph{ACM TKDD}, 2014.

\bibitem{wilson:eurosys09}
C.~Wilson, B.~Boe, A.~Sala, K.~P. Puttaswamy, and B.~Y. Zhao, ``User
  interactions in social networks and their implications,'' in \emph{Eurosys},
  2009.

\bibitem{gilbert:chi09}
E.~Gilbert and K.~Karahalios, ``Predicting tie strength with social media,'' in
  \emph{CHI}, 2009.

\bibitem{sybildefender}
W.~Wei, F.~Xu, C.~Tan, and Q.~Li, ``{SybilDefender}: Defend against {Sybil}
  attacks in large social networks,'' in \emph{IEEE INFOCOM}, 2012.

\bibitem{wang2010poisonedwater}
Y.~Wang and A.~Nakao, ``Poisonedwater: An improved approach for accurate
  reputation ranking in p2p networks,'' \emph{FGCS}, vol.~26, no.~8, pp.
  1317--1326, 2010.

\bibitem{alvisiSybil13}
L.~Alvisi, A.~Clement, A.~Epasto, S.~Lattanzi, and A.~Panconesi, ``Sok: The
  evolution of sybil defense via social networks,'' in \emph{IEEE S \& P},
  2013.

\bibitem{saad2003iterative}
Y.~Saad, \emph{Iterative methods for sparse linear systems}.\hskip 1em plus
  0.5em minus 0.4em\relax Siam, 2003.

\bibitem{derzko1965bounds}
N.~Derzko and A.~Pfeffer, ``Bounds for the spectral radius of a matrix,''
  \emph{Mathematics of Computation}, vol.~19, no.~89, 1965.

\bibitem{levin2009markov}
D.~A. Levin, Y.~Peres, and E.~L. Wilmer, \emph{Markov chains and mixing
  times}.\hskip 1em plus 0.5em minus 0.4em\relax American Mathematical Soc.,
  2009.

\bibitem{erdos1960evolution}
P.~Erdos and A.~R{\'e}nyi, ``On the evolution of random graphs,'' \emph{Publ.
  Math. Inst. Hung. Acad. Sci}, vol.~5, no.~1, pp. 17--60, 1960.

\bibitem{Barabasi99}
A.-L. Barab{\'a}si and R.~Albert, ``Emergence of scaling in random networks,''
  \emph{Science}, vol. 286, 1999.

\bibitem{kwak2010twitter}
H.~Kwak, C.~Lee, H.~Park, and S.~Moon, ``What is twitter, a social network or a
  news media?'' in \emph{Proceedings of the 19th international conference on
  World wide web}.\hskip 1em plus 0.5em minus 0.4em\relax ACM, 2010, pp.
  591--600.

\bibitem{Viswanath10}
B.~Viswanath, A.~Post, K.~P. Gummadi, and A.~Mislove, ``An analysis of social
  network-based {Sybil} defenses,'' in \emph{ACM SIGCOMM}, 2010.

\bibitem{sybilscar}
B.~Wang, L.~Zhang, and N.~Z. Gong, ``Sybilscar: Sybil detection in online
  social networks via local rule based propagation,'' in \emph{IEEE INFOCOM},
  2017.

\end{thebibliography}


%



\appendices

\section{Proof of Theorem~\ref{theorem_1}}
\label{app:theorem1}

We denote $\mathcal{Z}_u = {q_u \prod_{v \in \Gamma(u)} f_{vu} + (1-q_u) \prod_{v \in \Gamma(u)} (1-f_{vu})}$. 
Rewriting $p_u = \frac{1}{\mathcal{Z}_u} q_u \prod_{v \in \Gamma(u)} f_{vu}$ with the corresponding residual variables yields
\begin{equation*}
\small
\begin{split}
& 0.5+\hat{p}_{u} = \frac{1}{\mathcal{Z}_u} \big(0.5+\hat{q}_{u}\big) \prod_{v \in \Gamma(u)} \big(0.5 + \hat{f}_{vu} \big)  \\
& \Longrightarrow \ln (1+2\hat{p}_{u}) = -\ln \mathcal{Z}_u + \ln (1+2\hat{q}_{u}) + \sum_{v \in \Gamma(u)} \ln \big(0.5 + \hat{f}_{vu} \big) \\
& = -\ln \mathcal{Z}_u + \ln (1+2\hat{q}_{u}) + \sum_{v \in \Gamma(u)} \ln \big(0.5 \big) + \sum_{v \in \Gamma(u)} \ln (1 + 2 \hat{f}_{vu}) \\ 
\end{split}
\end{equation*}
Using approximation $\ln (1+x) \approx x$ when $x$ is small, we have: 
\begin{equation}
\label{approx_sybil}
2 \hat{p}_{u} = -\ln \mathcal{Z}_u  + 2\hat{q}_{u} +  |\Gamma(u)| \cdot \ln (0.5) + \sum_{v \in \Gamma(u)} 2 \hat{f}_{vu}.
\end{equation}

Similarly, via rewriting $1-p_u = \frac{1}{\mathcal{Z}_u} (1-q_u) \prod_{v \in \Gamma(u)} (1-f_{vu})$ with the corresponding residual variables and using approximation $\ln (1-x) \approx -x$ when $x$ is small, we have:
\begin{equation}
\label{approx_benign}
-2 \hat{p}_{u} = -\ln \mathcal{Z}_u -2\hat{q}_{u} + |\Gamma(u)| \cdot \ln (0.5) - \sum_{v \in \Gamma(u)} 2 \hat{f}_{vu}.
\end{equation}

Adding Equation~\ref{approx_sybil} with Equation~\ref{approx_benign} yields $\ln \mathcal{Z}_u = |\Gamma(u)| \cdot \ln (0.5)$. Via substituting this relation into Equation~\ref{approx_sybil} or Equation~\ref{approx_benign}, we have: 
\begin{equation}
\hat{p}_{u}  = \hat{q}_{u} + \sum_{v \in \Gamma(u)} \hat{f}_{vu}.
\end{equation}

\section{Proof of Theorem~\ref{theorem:bound}}
\label{app:analysis}

\myparatight{Overview} We leverage the classic analysis methods proposed by the authors of SybilRank. Specifically, SybilRank proposed the following three classic steps: 1) modeling the exchange of trust scores between benign region and Sybil region in each iteration, 2) modeling the trust score dynamics in the benign and Sybil regions, and 3) assuming the increased trust scores in the Sybil region all focus on a small group of Sybils. We follow these three steps to analyze residual posterior probabilities in the simplified version of SybilSCAR-D. However, one key difference is that SybilSCAR-D uses a different local rule with SybilRank, and thus the mathematical details in all the three steps are different.

\myparatight{Notations} We denote by $\mathcal{B}$ and $\mathcal{S}$ the set of benign nodes and Sybils, respectively. 
We denote by $d(\mathcal{B})$ and $d(\mathcal{S})$ the average degree of benign nodes and Sybil nodes, respectively.  
We denote by $|\mathcal{B}|$ and $|\mathcal{S}|$ the number of benign nodes and Sybils, respectively. 
For a node set $\mathcal{N}$, we denote its \emph{volume} as the sum of degrees of nodes in 
$\mathcal{N}$, i.e., $Vol(\mathcal{N})$ = $\sum_{u \in \mathcal{N}}$ $d_u$.
Moreover, we have 
 {\small
 \begin{align}
 & C_\mathcal{B} = \frac{g}{Vol(\mathcal{B})}, \quad
 C_\mathcal{S} = \frac{g}{Vol(\mathcal{S})},
\end{align}
}%
which were introduced by SybilRank.
We denote by $\hat{P}_\mathcal{B}^{(t)}$ and $\hat{P}_\mathcal{S}^{(t)}$ the average residual posterior probability of benign nodes and Sybils  in the $t${th} iteration, respectively. 
Initially, $\hat{P}_\mathcal{S}^{(0)} = 0$ (since we do not consider labeled Sybils in the training dataset) and $\hat{P}_\mathcal{B}^{(0)} < 0$.

\myparatight{Exchange of residual posterior probabilities between benign region and Sybil region}
In the $(t+1)${th} iteration, the average residual posterior probability of Sybils and the average residual posterior probability of benign nodes can be approximated as follows:
{\small
\begin{align}
\hat{P}_\mathcal{S}^{(t+1)}=C_\mathcal{S}\hat{P}_\mathcal{B}^{(t)}+ (1-C_\mathcal{S})\hat{P}_\mathcal{S}^{(t)}, \label{sfromto} \\
\hat{P}_\mathcal{B}^{(t+1)}=C_\mathcal{B}\hat{P}_\mathcal{S}^{(t)}+ (1-C_\mathcal{B})\hat{P}_\mathcal{B}^{(t)}. \label{bfromto}
\end{align}
}%
We take Equation~\ref{sfromto} as an example to illustrate how we derive the equations. In the considered version of SybilSCAR-D, in each iteration, a node's residual posterior probability is the average of its neighbors' residual posterior probabilities. Since we assume the attack edges are randomly established between benign nodes and Sybils, the total residual  posterior probability propagated from the benign nodes to Sybils is $g \hat{P}_\mathcal{B}^{(t)}$. Moreover, the total residual  posterior probability propagated within the Sybils is $(Vol(\mathcal{S})-g)\hat{P}_\mathcal{S}^{(t)}$, because each edge between Sybils contributes one copy of $\hat{P}_\mathcal{S}^{(t)}$ on average. Since each node takes the average of its neighbors' residual posterior probabilities, each Sybil has an average residual  posterior probability $\hat{P}_\mathcal{S}^{(t+1)}$ as $\frac{1}{|\mathcal{S}| d(\mathcal{S})}(g \hat{P}_\mathcal{B}^{(t)} + (Vol(\mathcal{S})-g)\hat{P}_\mathcal{S}^{(t)})$, which gives us Equation~\ref{sfromto}.  

Note that the derivation of Equations~\ref{sfromto} and~\ref{bfromto} is inspired by SybilRank~\cite{sybilrank}. However, since SybilSCAR-D and SybilRank use different local rules (though both are linear local rules), their exchange dynamics between benign region and Sybil region are different. This difference leads to different dynamics within benign/Sybil region and eventually leads to different security guarantees. 
 Given Equations~\ref{sfromto} and~\ref{bfromto}, we have: 
{\small 
\begin{align}
\hat{P}_\mathcal{B}^{(t+1)}-\hat{P}_\mathcal{S}^{(t+1)} = (1-C_\mathcal{B}-C_\mathcal{S})(\hat{P}_\mathcal{B}^{(t)}-\hat{P}_\mathcal{S}^{(t)}). 		
\end{align}
}%

\myparatight{Dynamics in the Sybil and benign regions} The decrease of the average residual posterior probabilities of Sybils is as follows:
{\small 
\begin{align}
&\hat{P}_\mathcal{S}^{(t+1)}-\hat{P}_\mathcal{S}^{(t)} = C_\mathcal{S}(\hat{P}_\mathcal{B}^{(t)}-\hat{P}_\mathcal{S}^{(t)}) \\
& = C_\mathcal{S} (1-C_\mathcal{B}-C_\mathcal{S})(\hat{P}_\mathcal{B}^{(t)}-\hat{P}_\mathcal{S}^{(t)}) \\
& =(1-C_\mathcal{B}-C_\mathcal{S})^{t}C_\mathcal{S}(\hat{P}_\mathcal{B}^{(0)}-\hat{P}_\mathcal{S}^{(0)}),
\end{align}
}%
where the above equation is negative (so we call it a decrease) because $(\hat{P}_\mathcal{B}^{(0)}-\hat{P}_\mathcal{S}^{(0)})$ is negative. 
Therefore, we have:
{\small
\begin{align}
&\hat{P}_\mathcal{S}^{(t)}-\hat{P}_\mathcal{S}^{(0)} =\sum_{i=0}^{t-1}(1 - C_\mathcal{B} - C_\mathcal{S})^{t} \times C_\mathcal{S} (\hat{P}_\mathcal{B}^{(0)}-\hat{P}_\mathcal{S}^{(0)}).
\end{align}
}%
Similarly, the increase of the average residual posterior probabilities of benign nodes is as follows:
{\small
\begin{align}
&\hat{P}_\mathcal{B}^{(t+1)}-\hat{P}_\mathcal{B}^{(t)} = -C_\mathcal{B}(\hat{P}_\mathcal{B}^{(t)}-\hat{P}_\mathcal{S}^{(t)}) \\
& =-(1-C_\mathcal{B}-C_\mathcal{S})^{t}C_\mathcal{B} (\hat{P}_\mathcal{B}^{(0)}-\hat{P}_\mathcal{S}^{(0)}) \\
& =-(1 - C_\mathcal{B} - C_\mathcal{S})^{t} \times C_\mathcal{B} (\hat{P}_\mathcal{B}^{(0)}-\hat{P}_\mathcal{S}^{(0)}),
\end{align}
}%
where the above equation is positive (so we call it an increase) because $(\hat{P}_\mathcal{B}^{(0)}-\hat{P}_\mathcal{S}^{(0)})$ is negative. 
Furthermore, we have:
{\small
\begin{align}
&\hat{P}_\mathcal{B}^{(t)}-\hat{P}_\mathcal{B}^{(0)} =-\sum_{i=0}^{t-1}(1 - C_\mathcal{B} - C_\mathcal{S})^{t} \times C_\mathcal{B} (\hat{P}_\mathcal{B}^{(0)}-\hat{P}_\mathcal{S}^{(0)}).
\end{align}
}%

\myparatight{Security guarantee} We assume after $\Omega=O(\log |V|)$ iterations, benign nodes have similar residual posterior probabilities, which are the average residual posterior probability of benign nodes. We note that SybilRank relies on a similar assumption, i.e., after $O(\log |V|)$ iterations, benign nodes have similar degree-normalized trust scores. 
We assume the decrease of residual posterior probabilities of Sybils all focus on $n_{\mathcal{S}}$ Sybils, which gives an upper bound of Sybils whose residual posterior probabilities are smaller than benign nodes. 
If we want these Sybils to have residual posterior probabilities that are smaller than benign nodes, then we have:
{\small
\begin{align}
0 - \frac{(\hat{P}_\mathcal{S}^{(0)}-\hat{P}_\mathcal{S}^{(\Omega)})|\mathcal{S}|}{n_{\mathcal{S}}} < \hat{P}_\mathcal{B}^{(\Omega)}\\
\iff n_{\mathcal{S}} < \frac{(\hat{P}_\mathcal{S}^{(\Omega)}-\hat{P}_\mathcal{S}^{(0)})|\mathcal{S}|}{\hat{P}_\mathcal{B}^{(\Omega)}-0} \\
\iff n_{\mathcal{S}} < \frac{(\hat{P}_\mathcal{S}^{(\Omega)}-\hat{P}_\mathcal{S}^{(0)})|\mathcal{S}|}{\hat{P}_\mathcal{B}^{(\Omega)}-\hat{P}_\mathcal{S}^{(0)}},
\end{align}
}%
where the last step holds because $\hat{P}_\mathcal{S}^{(0)}=0$.
Moreover, we have:
{\footnotesize
\begin{align}
& \frac{(\hat{P}_\mathcal{S}^{(\Omega)}-\hat{P}_\mathcal{S}^{(0)}) |\mathcal{S}|}{\hat{P}_\mathcal{B}^{(\Omega)}-\hat{P}_\mathcal{S}^{(0)}} =\frac{(\hat{P}_\mathcal{S}^{(\Omega)}-\hat{P}_\mathcal{S}^{(0)})|\mathcal{S}|}{\hat{P}_\mathcal{B}^{(\Omega)}-\hat{P}_\mathcal{B}^{(0)}+\hat{P}_\mathcal{B}^{(0)}-\hat{P}_\mathcal{S}^{(0)}} \label{2} \\
 &= \frac{\sum_{0\leq t\leq (\Omega-1)}(1-C_\mathcal{S}-C_\mathcal{B})^{t}C_\mathcal{S} (\hat{P}_\mathcal{B}^{(0)}-\hat{P}_\mathcal{S}^{(0)})|\mathcal{S}|}{(1-\sum_{0\leq t\leq (\Omega-1)}(1-C_\mathcal{S}-C_\mathcal{B})^{t}C_\mathcal{B})(\hat{P}_\mathcal{B}^{(0)}-\hat{P}_\mathcal{S}^{(0)})} \label{3} \\
 &<\frac{\sum_{0\leq t\leq (\Omega-1)}(1-C_\mathcal{S})^{t}C_\mathcal{S}|\mathcal{S}|}{(1-\sum_{0\leq t\leq (\Omega-1)}(1-C_\mathcal{S}-C_\mathcal{B})^{t}C_\mathcal{B})} \\
 &=\frac{(1-(1-C_\mathcal{S})^{\Omega})|\mathcal{S}|}{1-\frac{1-(1-C_\mathcal{S}-C_\mathcal{B})^\Omega}{C_\mathcal{S}+C_\mathcal{B}}C_\mathcal{B}} \label{5} \\
&\approx \frac{C_\mathcal{S}\Omega |\mathcal{S}|}{1-\frac{1-(1-(C_\mathcal{S}+C_\mathcal{B})\Omega+\frac{\Omega(\Omega-1)}{2}(C_\mathcal{S}+C_\mathcal{B})^2)}{C_\mathcal{S}+C_\mathcal{B}}C_\mathcal{B}} \label{6} \\
&\approx \frac{g \Omega}{d(\mathcal{S})(1-(\Omega-\frac{\Omega^2}{2}(C_\mathcal{S}+C_\mathcal{B}))C_\mathcal{B})} \label{7} \\
&< \frac{g \Omega}{d(\mathcal{S})(1-(\Omega-\frac{\Omega^2}{2}(C_\mathcal{S}+C_\mathcal{B}))(C_\mathcal{S}+C_\mathcal{B}))} \label{8} \\
&= \frac{g \Omega}{d(\mathcal{S})(\frac{1}{2}+\frac{1}{2}(\Omega(C_\mathcal{S}+C_\mathcal{B})-1)^2)} \label{9} \\
&\leq \frac{2 g \Omega}{d(\mathcal{S})}, \label{10}
\end{align}
}%
where in Equation~\ref{6}, we use the first-order and second-order Taylor expansion on nominator and denominator, respectively; 
In Equation~\ref{7}, we use $C_\mathcal{S} = \frac{g}{Vol(\mathcal{S})}$, $Vol(\mathcal{S}) = d(\mathcal{S}) |\mathcal{S}|$, and $ \Omega(\Omega-1)(C_\mathcal{S}+C_\mathcal{B})$ $\approx$ $ \Omega^2 (C_\mathcal{S}+C_\mathcal{B})$. 
By setting $\Omega=O(\log |V|)$, we have:
\begin{align}
 n_{\mathcal{S}} = O\Big(\frac{g\log |V|}{d(\mathcal{S})}\Big).
\end{align}

\end{document}